\documentclass[twocolumn,prl,superscriptaddress,longbibliography,footinbib]{revtex4-2}
\usepackage[dvips]{graphicx} 
\usepackage{amsfonts,amscd,amsmath,amsthm}
\usepackage{enumerate}
\usepackage{epsfig}
\usepackage{subfigure}
\usepackage{xcolor}
\usepackage[colorlinks = true]{hyperref}
\usepackage{physics}
\usepackage{epstopdf}
\usepackage{framed}
\usepackage{multirow}

\usepackage{longtable}
\usepackage{comment}
\usepackage[ruled,vlined]{algorithm2e}
\usepackage[most]{tcolorbox}
\usepackage{braket}
\graphicspath{{./figure/}}

\usetikzlibrary{tikzmark, calc, fit, positioning}
\usetikzlibrary{shapes}
\usepackage{pgfplots}
\usepgfplotslibrary{groupplots}
\usetikzlibrary{calc}          
\pgfplotsset{compat=1.18}

\newtheorem{theorem}{Theorem}
\newtheorem{lemma}{Lemma}

\newtheorem{definition}{Definition}
\newtheorem{proposition}{Proposition}

\newtcolorbox[auto counter]{mybox}[2][]{
	enhanced,
	breakable,
	colback=blue!5!white,
	colframe=blue!75!black,
	fonttitle=\bfseries,
	title=Box \thetcbcounter: #2,#1
}

\begin{document}

\title{Detecting non-Gaussianity and quantum non-Gaussianity via multicopy interference}

\author{Hao Dai}
\email{dhao@bimsa.cn}
\affiliation{Beijing Institute of Mathematical Sciences and Applications, Beijing 101408, China}
\author{Matteo G. A. Paris}
\email{matteo.paris@fisica.unimi.it}
\affiliation{Dipartimento di Fisica, Università di Milano, 20133 Milan, Italy}
\author{Yue Zhang}
\email{Corresponding author: zhangyue115@amss.ac.cn}
\affiliation{State Key Laboratory of Mathematical Sciences, Academy of Mathematics and Systems Science, Chinese Academy of Sciences, Beijing 100190, China}
\affiliation{School of Mathematical Sciences, University of Chinese Academy of Sciences, Beijing 100049, China}

\begin{abstract}
We develop tests of non-Gaussianity for bosonic systems based on interference between identically prepared copies of the state. The tests use real orthogonal transformations that mix the copies while preserving their collective mode. These transformations commute with identical Gaussian unitaries applied to each copy, allowing us to derive exact Gaussian benchmarks. For single-mode states, purity alone determines the Gaussian reference value, and any discrepancy certifies non-Gaussianity. A two-copy SWAP measurement supplies the purity calibration, while informative tests within this family begin at three copies. We describe a measurement protocol with finite-sample error bounds that use passive interferometry and photon-number-resolving detectors. We also derive purity-dependent inequalities, including joint constraints on higher-power real moments, whose violation excludes convex mixtures of Gaussian states (thus certifying quantum non-Gaussianity). Finally, we extend the reference formulas to multimode systems and prove completeness of the covariance-calibrated hierarchy for states with finite mean photon number.
\end{abstract}

\maketitle

\textit{Introduction}---Gaussian states are central to continuous-variable quantum information because they are experimentally accessible and admit an efficient phase-space description \cite{Bach2004,Biam2017,Will1936,Cush1971,Hole1975,Hole1982,Gros2003,Brau2005,Pira2006,Fern2008,DAur2009,Gior2010,Daem2010,Weed2012}. Non-Gaussian states are however needed for several tasks beyond the reach of Gaussian resources alone and are especially important in fault-tolerant continuous-variable quantum computation \cite{Dell2010,Cerf2005,Xu2022,PhysRevAGottesman2001,Hahn2025PRX}. This has motivated resource-theoretic treatments of non-Gaussianity \cite{Alle2010,Alle2012,Saba2017,Taka2018,Alba2018,Chit2019,Chabaud21,Fiurasek21,Lachman22, Chat2026,Hotter26}.
Detecting non-Gaussianity is in fact nontrivial: the set of Gaussian states is not convex, and a state may be non-Gaussian even when its Wigner function is nonnegative. Existing criteria use phase-space distributions, information-theoretic quantities, uncertainty relations, and distances from suitable Gaussian references \cite{Geno2010,Geno2013,Mari2013,Ghiu2013,Hugh2014,Son2015,Baek2018,Fu20201,Zhan2020,Park2021,Malp2023,Turn2025}. Here we take a different route and use symmetries of identically prepared copies to design interferometric tests for non-Gaussianity and quantum non-Gaussianity.

The key observation is that suitably chosen interference measurements on identical copies are insensitive to the displacement and squeezing used to prepare a Gaussian state. To see why, recall that every single-mode Gaussian state can be obtained from a thermal state by a Gaussian unitary. Passive interferometers that mix the copies with real orthogonal coefficients commute with identical squeezing and phase rotations applied to each copy. They also commute with a common displacement provided they leave unchanged the collective, equal-weight, sum of the copy modes, in which that displacement is concentrated. For these interferometers, the expectation value of the corresponding unitary is therefore the same as for identical copies of the underlying thermal state. Since the thermal parameter is uniquely determined by the single-mode purity, an independent purity measurement fixes the Gaussian prediction without requiring knowledge of the displacement or squeezing. Any departure from this prediction certifies non-Gaussianity.

We derive an exact Gaussian reference formula and show that, for single-mode states with finite mean photon number, the full family of tests is complete as the copy number varies. We also derive a purity-dependent witness for quantum non-Gaussianity (i.e., for states that cannot be expressed as a convex combination of pure Gaussian state) and joint constraints on real moments. For three copies, the witness can be measured with a direct finite-sample error bound by using an inverse Fourier tritter and photon-number-resolved detection. Finally, we extend the analysis to multimode states, where the Gaussian reference is determined by the Williamson thermal parameters.

\textit{Copy-mixing symmetry}---Consider a $d$-mode bosonic system with quadratures $q_j=(a_j+a_j^\dagger)/\sqrt2$ and $p_j=(a_j-a_j^\dagger)/(i\sqrt2)$, collected into $\boldsymbol R=(q_1,p_1,\ldots,q_d,p_d)^T$. A Gaussian state has a Gaussian Wigner function and is completely specified by its mean vector $\boldsymbol d=\langle \boldsymbol R\rangle={\rm Tr} (\rho \boldsymbol{R}) $ and covariance matrix (CM) $\sigma$ with entries $\sigma_{j\ell}=\langle\{R_j-d_j,R_\ell-d_\ell\}\rangle$, where $d_j=\langle R_j\rangle$ and $\{A,B\}=AB+BA$ is the anticommutator. With this convention, the vacuum CM is $I_{2d}$. Gaussian unitaries shift the center and reshape the phase-space distribution through displacement, squeezing, and rotation \cite{adesso2014continuous}. The displacement
$D(\boldsymbol{\alpha})
    =\exp\!\big( \sum_{j=1}^d  ( \alpha_j a_j^\dagger-\overline{\alpha_j}a_j)
    \big)
$
shifts the mean vector for any $\boldsymbol{\alpha}\in\mathbb{C}^d.$ A symplectic matrix $S\in\operatorname{Sp}(2d,\mathbb R)$ is implemented by $\Lambda(S)$ through 
 $\Lambda(S)^\dagger\boldsymbol R\Lambda(S)=S\boldsymbol R$,
which sends the mean vector and CM to $S\boldsymbol d$ and $S\sigma S^T$ respectively, while preserving the canonical commutation relations. Up to an overall phase, a general Gaussian unitary can be written as $U=D(\boldsymbol\alpha)\Lambda(S).$

We now consider the interference of $k$ independently prepared, identical copies of a single-mode bosonic state. Let 
$
 \boldsymbol a=(a_1,\ldots,a_k)^T
$
collect their mode operators. Any passive interferometer corresponds to real orthogonal matrix $O\in O(k)$ with mode transformation $\Gamma(O)$ 
$
 \Gamma(O)\boldsymbol a\Gamma(O)^{-1}=O^T \boldsymbol a.
$
We use the vacuum-preserving convention $\Gamma(O)\ket{0}^{\otimes k}=\ket{0}^{\otimes k}$. Those 
interferometers conserve total photon number and mix the position and momentum quadratures with the same real coefficients. The notation $\Gamma$ also applies to complex unitary mode-mixing matrices, including the Fourier tritter used below. Real orthogonal matrices form the subgroup relevant to our symmetry argument.

For every orthogonal matrix $O\in SO(k)$, one can write $O=e^A$, where $A$ is a real antisymmetric generator {\color{black} with bosonic representation
\cite{woit2017quantum},
 $\Gamma'(A)=\sum_{j,\ell=1}^{k}
 A_{j\ell}a_j^\dagger a_\ell$,
 $A^T=-A$,
and $\Gamma(e^A)=e^{\Gamma'(A)}$. Details are given in the Supplemental Material (SM) \cite{SupplementalMaterial}.} We seek mixers whose witness expectations are independent of the common displacement and symplectic transformation applied to all copies. Since real orthogonal mixing acts identically on the position and momentum quadratures, it commutes with an identical symplectic transformation on every copy. This gives the following lemma.
\begin{lemma}\label{th:com_sym}
For every \(O\in O(k)\) and \(S\in Sp(2,\mathbb{R})\), 
  $    \left[        \Gamma(O),\Lambda(S)^{\otimes k}
    \right]=0$.
\end{lemma}
{\color{black} \noindent Lemma~\ref{th:com_sym} says that passive beam-splitter network commutes 
with identical squeezing/rotation on every copy and thus is insensitive to any 
identical symplectic transformation applied to all copies.} A common displacement, however, imposes an additional constraint. Equal displacements act only on the collective mode, 
and leave the $k-1$ orthogonal relative modes undisplaced. Since the collective mode 
corresponds to the direction
$\boldsymbol e_1=k^{-1/2}(1,\ldots,1)^T$ in the space of copies, an interferometer 
preserves a common displacement precisely when it leaves this collective direction 
invariant.
\begin{lemma}\label{lem:dis_com}
Let \(O\in O(k)\). Then
$    [\Gamma(O),D(\alpha)^{\otimes k}]=0,
    \, \forall\,  \alpha\in\mathbb{C}$,
if and only if $O\boldsymbol{e}_1=\boldsymbol{e}_1$, i.e., 
if the interferometer leaves the symmetric (common) input port unchanged.
\end{lemma}
\noindent The two lemmas together identify the transformation/interferometers 
compatible with a common Gaussian preparation (from now on, admissible mixers).
\begin{theorem}\label{th:main}
For $\alpha\ne0$ and $S\in\operatorname{Sp}(2,\mathbb R)$, an orthogonal copy mixer $\Gamma(O)$ commutes with $\big(D(\alpha\big)\Lambda(S))^{\otimes k}$ if and only if 
$ O\boldsymbol e_1=\boldsymbol e_1,
 \,O\in O(k)$.
The admissible mixers form the subgroup $\widetilde G_k\cong O(k-1)$, i.e., they  act only on the $k-1$ relative (antisymmetric) modes.
\end{theorem}
\noindent
The choice of mixer is independent
of the displacement and squeezing. The condition in Theorem \ref{th:main}
remains sufficient for $\alpha=0$. 

\textit{Single-Mode Case}---Every single-mode Gaussian state can be prepared from a thermal state 
$ \rho_\lambda=(1-\lambda)\sum_{n=0}^{\infty} \lambda^n\ket{n}\bra{ n}$ by squeezing, rotation, and displacement, that is, $ \rho_g=D(\alpha)\Lambda(S)\rho_\lambda
 \Lambda(S)^\dagger D(\alpha)^\dagger.
$ For $O\in\widetilde G_k$, the commutation relations established above imply
${\rm Tr}\big(\rho_g^{\otimes k}\Gamma(O)\big)
={\rm Tr}\big(\rho_\lambda^{\otimes k}\Gamma(O)\big).$
Thus, for fixed $k$ and $O$, the expectation value is the same for all Gaussian states with the same purity. The thermal parameter $\lambda$ is uniquely determined by the purity
$p_{\rho_g}={\rm Tr}(\rho_g^2)$ as $\lambda=(1-p_{\rho_g})/(1+p_{\rho_g})$.
\begin{figure}[h!]
  \centering \includegraphics[width=0.50\textwidth]{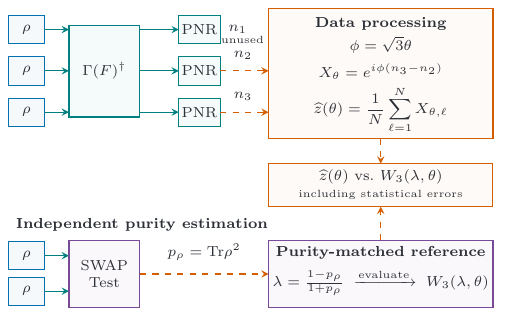}
\caption{Three-copy measurement protocol. Three identical copies $\rho$ pass through an inverse Fourier tritter $\Gamma(F)^\dagger$ and photon-number-resolved (PNR) detectors. Only $n_2,n_3$ enter $X_\theta=e^{i\sqrt{3}\theta(n_3-n_2)}$, whose sample average estimates $z(\theta)$. An independent two-copy SWAP measurement, implemented with a balanced beam splitter and photon-number parity detection in the relative output port, estimates $p_\rho$ and hence $\lambda=(1-p_\rho)/(1+p_\rho)$ for the purity-matched Gaussian reference. A discrepancy exceeding the combined statistical uncertainty certifies non-Gaussianity. Solid lines denote optical paths; dashed arrows carry classical data.}
\label{fig:1}
\end{figure}

Now, for an arbitrary state $\rho$, we define $k$-copy expectation
$f_k(\rho,O)={\rm Tr}\big(\rho^{\otimes k}\Gamma(O)\big), \,O\in\widetilde G_k$,
and its Gaussian reference value 
$W_k(\lambda,O)={\rm Tr}\big(\rho_g^{\otimes k}\Gamma(O)\big)={\rm Tr}\big(\rho_\lambda^{\otimes k}\Gamma(O)\big)$,
where ${\rm Tr}(\rho_\lambda^2)={\rm Tr}(\rho^2)$.  We have
{\color{black} 
\begin{theorem}[Gaussian reference]\label{th:refval}
For $\lambda\in[0,1)$, $O\in\widetilde G_k$, and $I_k$ being the identity operator.
\begin{equation}\label{eq:reference}
 W_k(\lambda,O)= \frac{(1-\lambda)^k}{\det(I_k-\lambda O)}.
\end{equation}
\end{theorem}}
The Gaussianity test for a state $\rho$ 
proceed as follows: choose a thermal reference $\rho_\lambda$ with the same purity of $\rho$
(which itself can be estimated by a two-copy SWAP measurement \cite{PhysRevLett.107.170505}, see Fig.~\ref{fig:1}). If 
$\rho$ is Gaussian, $f_k(\rho,O)=W_k(\lambda,O)$ for every
admissible $O$. Any discrepancy for any given $O$ thus certifies the non-Gaussianity.
{\color{black} To derive the value of  $W_k(\lambda,O)$ in Eq.~\eqref{eq:reference}, write} $O=V\operatorname{diag} (e^{i\phi_1},\ldots,e^{i\phi_k})V^\dagger$, where $V$ is unitary and $\phi_j\in\mathbb R$. In the corresponding normal-mode basis, the interferometer acts as independent phase rotations. An occupation of $n$ photons in mode $j$ contributes a phase factor $e^{in\phi_j}$. Averaging over the thermal photon-number distribution yields $(1-\lambda)/(1-\lambda e^{i\phi_j})$ for each mode. Multiplying these factors gives Theorem~\ref{th:refval}. Since $O$ is real, $\det(I_k-\lambda O)$ and the reference value are real. For $k=2$, the only admissible mixers are the identity and SWAP, whose expectation values are $1$ and
$p_\rho=\operatorname{Tr}(\rho^2)$, respectively, for any state $\rho$. Thus, nontrivial tests within this purity-calibrated single-mode family require at least three copies.

\textit{Examples}---A basis of infinitesimal generators for general $\widetilde G_k$ is given in the SM (Sec.B.5) \cite{SupplementalMaterial}. Here we consider $k=3$, for which $\widetilde G_3\cong O(2)$ acts on the two-dimensional space of relative modes. 
{\color{black} Its connected component $SO(2)$ consists of planar rotations specified 
by a single angle, such that
$     \Gamma(e^{\theta G})=\exp \big(\theta(
    a_1 a_2^\dagger-a_1^\dagger a_2
    +a_3 a_1^\dagger-a_1 a_3^\dagger
    +a_2 a_3^\dagger-a_2^\dagger a_3) \big)$.
Let us now denote $f_3(\rho,\theta)=f_3(\rho,e^{\theta G}),$ and similarly, $W_3(\rho,\theta)=W_3(\rho,e^{\theta G})$. For a given value of the purity, i.e., of $\lambda$, 
the Gaussian reference value is given by
\begin{equation} \label{W_3}
W_3(\lambda,\theta) = \frac{(1-\lambda)^2} {1-2\lambda\cos(\sqrt{3}\theta)+\lambda^2}. 
\end{equation} }

Example 1. For any Fock state \(\ket{n}\) with \(n\geq 1\), direct calculation gives that the second $\theta$-derivative of the three-copy expectation at $0$ is \(f_3''(\ket{n}\bra{ n},0) = -6n(n+1) \neq 0, \) which indicates that \(f_3(\ket{n}\bra{ n},\theta)\) is not constant. Since every Fock state is pure, the corresponding Gaussian reference parameter is \(\lambda=0\), for which $W_3(0,\theta)=1 $ for all \(\theta\). Hence \(f_3(\ket{n}\bra{ n},\theta)\) cannot coincide with the Gaussian reference value, certifying that every Fock state with \(n\geq1\) is non-Gaussian. 

Example 2. For the even cat state $\ket{C_+} \propto  \left( \ket{\alpha}+\ket{-\alpha} \right)$, with $\alpha\neq0$, we have
$f_3''(\ket{C_+}\bra{ C_+},0) = -6|\alpha|^2 ( \tanh |\alpha|^2 - |\alpha|^2\sech^2 |\alpha|^2 )$,
whereas the corresponding Gaussian reference for this pure state satisfies \(W_3(0,\theta)=1\). Therefore, every nontrivial even cat state is certified to be non-Gaussian. 

Example 3. As a mixed-state example, consider the incoherent mixture of two coherent states
$\rho_{\mathrm{mix}} =  \left( \ket{\alpha}\bra{\alpha} + \ket{-\alpha}\bra{-\alpha} \right)/2,\,\alpha\neq0.$ Unlike the even cat state, \(\rho_{\mathrm{mix}}\) is an incoherent statistical mixture of the two coherent states. Writing $r=|\alpha|^2$, the three-copy expectation is $f_3(\rho_{\mathrm{mix}},\theta)=\Big(1+3e^{-8r\big(1-\cos(\sqrt3\theta)\big){3}}\Big)/4$. Fig.~\ref{fig:2} shows the difference $\Delta(r,\theta)=f_3(\rho_{\mathrm{mix}},\theta)
-W_3(\lambda(r),\theta)$ for the two settings $\theta=\pi/2$ and $2\pi/(3\sqrt3)$. Both yield $\Delta<0$ for any $r>0$, so either setting alone certifies ordinary non-Gaussianity for any nonzero amplitude in the exact-data limit.

\begin{figure}[h]
    \centering \includegraphics[width=0.45\textwidth]{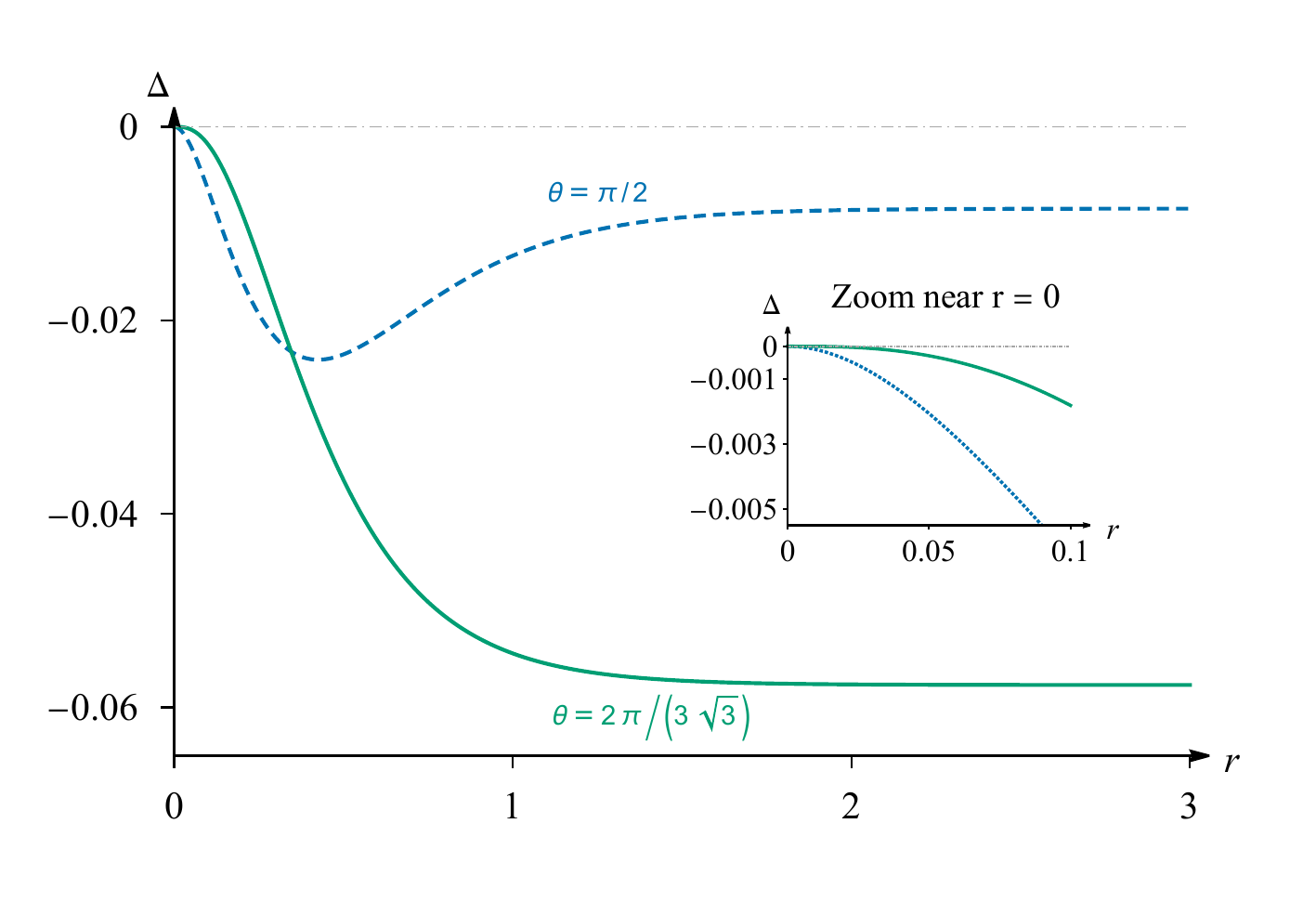}
\caption{Difference $\Delta=f_3(\rho_{\mathrm{mix}},\theta)
-W_3(\lambda(r),\theta)$ for the coherent-state mixture, with $r=|\alpha|^2$. The dashed blue and solid green curves correspond to $\theta=\pi/2$ and
$2\pi/(3\sqrt3)$, respectively. The inset enlarges the region $0< r\leq0.1$. The dash dotted lines mark $\Delta=0$.}
    \label{fig:2}
\end{figure}

{\color{black} \textit{Measurement scheme and statistical analysis}---We illustrate the measurement scheme for the three-copy witness. For the general protocol see F.1--F.2 of SM \cite{SupplementalMaterial}. The orthogonal operation given above admits the decomposition
$\Gamma(e^{\theta G}) = \Gamma(F_3) \exp\!\big( i\sqrt{3}\theta (\hat n_3-\hat n_2) \big) \Gamma(F_3)^\dagger$,
where $F_3$ is the discrete Fourier transform \footnote{
\begin{equation*}
    F_3=
    \frac{1}{\sqrt{3}}
    \begin{pmatrix}
        1&1&1\\
        1&\omega&\omega^2\\
        1&\omega^2&\omega
    \end{pmatrix},
    \qquad
    \omega=e^{2\pi i/3}
\end{equation*}
}.
$\Gamma(F_3)$ can be implemented by a Fourier tritter \cite{PhysRevLett.73.58Reck1994,spagnolo2013three}, and photon counting at the output may be implemented directly or by photon chopping  \cite{Paul1996,Paul1997}. The three-copy 
expectation can be written as
$z(\theta) := {\rm Tr}\!\big(\Gamma(e^{\theta G})\rho^{\otimes 3} \big) = {\rm Tr}\!\big(e^{i\sqrt{3}\theta(\hat n_3-\hat n_2)} \Gamma(F_3)^\dagger \rho^{\otimes 3} \Gamma(F_3) \big)$,
so that the measurement is reduced to applying the inverse Fourier tritter $ \Gamma(F_3)^\dagger$ to the three copies, followed by photon-number-resolved (PNR) detection}. An independent two-copy SWAP measurement determines the purity, and hence fixes the parameter $\lambda$ of the purity-matched thermal reference and the Gaussian reference value given by Eq.~(\ref{W_3}). Therefore, if the deviation between \(\widehat{z}(\theta)\) and \(W_3(\widehat{\lambda},\theta)\) exceeds the statistical error, the state \(\rho\) is certified to be non-Gaussian. The protocol is illustrated in Fig.~\ref{fig:1}.

To account for finite-sample errors, let \(\epsilon_z\) and \(\epsilon_\lambda\) denote the additive errors in the tritter estimate \(\widehat z(\theta)\) and the purity-calibrated parameter \(\widehat\lambda\), with failure probabilities at most \(\delta_z\) and \(\delta_p\), respectively, and let \(L_\lambda\) be a Lipschitz constant
of \(W_3(\lambda,\theta)\) in \(\lambda\). If
 $ \left|\widehat z(\theta)-W_3(\widehat\lambda,\theta)\right|
    >\epsilon_z+L_\lambda\epsilon_\lambda$,
the state is certified as non-Gaussian with confidence at least \(1-\delta_z-\delta_p\). {\color{black} Since each shot yields a phase of unit modulus, Hoeffding's inequality 
gives the sample complexity \(O(\epsilon_z^{-2}\log\delta_z^{-1})\), i.e., a sufficient number of measurements is given by 
$
 M_{s}\geq \left\lceil 4/\epsilon_z^2\, \log 4/\delta_z\right\rceil
$.}

\textit{Multimode Generalization}---The protocol can be extended to $d$ modes. Let $a_{r,\ell}$ 
denote the annihilation operator of mode $r=1,\ldots,d$ in copy $\ell=1,\ldots,k$. Applying the same orthogonal transformation $O$ to each mode gives
 $\Gamma_d(O)=\bigotimes_{r=1}^{d}\Gamma^{(r)}(O)$,
where $\Gamma^{(r)}(O)$ applies $O$ to the $k$-th copy of mode $r$. These transformations 
commute because they act on disjoint sets of modes. Using the same $O$ for every mode ensures commutation with identical multimode symplectic transformations, including those coupling different  modes. As in the single-mode case, a common displacement requires the collective copy direction to remain fixed. Specifically, for $U=D(\boldsymbol\alpha)\Lambda(S)$ with $S\in\operatorname{Sp}(2d,\mathbb R)$ and $\boldsymbol\alpha\ne 0$, an operator $\Gamma_d(O)$ satisfies
 $[\Gamma_d(O),\ U^{\otimes k} ]=0$,
if and only if $O\boldsymbol e_1=\boldsymbol e_1$ for any $O\in O(k).$ The admissible mixers again form the subgroup $\widetilde{G}_k$.  By Williamson's decomposition \cite{Will1936,Weed2012}, every $d$-mode Gaussian state can be written as
$\rho_g = D(\boldsymbol{\alpha})\Lambda(S) \rho_{\boldsymbol{\lambda}} \Lambda(S)^\dagger D(\boldsymbol{\alpha})^\dagger$, 
where $\rho_{\boldsymbol{\lambda}} = \rho_{\lambda_1} \otimes\cdots\otimes \rho_{\lambda_d}$ is a product of single-mode thermal states, with the thermal parameter vector \(\boldsymbol{\lambda}=(\lambda_1,\ldots,\lambda_d)\) given by the Williamson normal form of the CM. Consequently, the copy-mixing symmetry removes the Gaussian preparation from the reference expectation, leaving only the thermal state. 

\begin{theorem}[Multimode generalization]\label{th:multiref}
For any \(O\in \widetilde{G}_k\), the multimode Gaussian reference value is 
\begin{equation}\label{eq:multi_gaussian_reference} 
W_k(\boldsymbol{\lambda},O) := {\rm Tr}\!\left(   \rho_{\boldsymbol{\lambda}}^{\otimes k}\Gamma_d(O) \right) = \prod_{r=1}^{d} \frac{(1-\lambda_r)^k} {\det(I_k-\lambda_r O)}.
\end{equation}
\end{theorem}

For any Gaussian state $\rho_g$, ${\rm Tr}\big(\rho_g^{\otimes k}\Gamma_d(O)\big)=W_k(\boldsymbol\lambda,O)$ for every admissible mixer. This provides a non-Gaussianity criterion for an arbitrary $d$-mode state $\rho$ with finite total mean photon number: if there exists $O\in\widetilde G_k$ such that
 ${\rm Tr}\!\big(\rho^{\otimes k}\Gamma_d(O)\big) \ne W_k(\boldsymbol\lambda,O)$,
then \(\rho\) is certified to be non-Gaussian. 
For an arbitrary $d$-mode input state $\rho$ with finite quadrature second moments, let $\nu_1,\ldots,\nu_d$ denote the symplectic eigenvalues of its CM $\sigma_\rho$, and set $\lambda_r=(\nu_r-1)/(\nu_r+1)$. These parameters determine the covariance-matched Gaussian reference values $W_k(\boldsymbol\lambda,O)$, without requiring a separate purity measurement or assuming that $\rho$ is Gaussian. As proven in SM (Sec. C.3) \cite{SupplementalMaterial}, the two-mode Fock state \(\ket{1,1}\) can be certified to be non-Gaussian using this method.

Let us now establish a complete covariance-calibrated hierarchy for $d$-mode states with finite total mean photon number. Define $f_k^{(d)}(\rho,O):=\operatorname{Tr}(\rho^{\otimes k}\Gamma_d(O))$, so that $f_k^{(1)}=f_k$. The states in this class that remain undetected 
at the $k$-copy level form the set
 $    \mathcal{S}_k
    =
    \{\rho:
        f_k^{(d)}(\rho,O)=W_k(\boldsymbol{\lambda},O),
        \ \forall O\in\widetilde{G}_k
    \}$,
where \(\boldsymbol{\lambda}\) denotes the covariance-matched thermal reference parameters. Thus states in the complement \(\mathcal{S}_k^c\) are already certified as non-Gaussian.

\begin{theorem}[Completeness]\label{th:completeness}
Within the class of states with finite mean photon number (i.e. finite quadrature second moments), the sets satisfy $\mathcal{S}_{k+1}\subseteq\mathcal{S}_k$, with $\bigcap_{k=1}^{\infty}\mathcal{S}_k=\mathcal{G}_d$,
and where \(\mathcal{G}_d\) is the set of \(d\)-mode Gaussian states. Equivalently, \(\mathcal{S}_k^c\subseteq\mathcal{S}_{k+1}^c\), that is, once a state is detected as non-Gaussian at some copy number, it remains detected for all larger copy numbers. Consequently, every fixed non-Gaussian state is detectable by an admissible orthogonal witness at some finite copy number.
\end{theorem}

{\color{black} \textit{Quantum non-Gaussianity}---A state $\rho$ is quantum non-Gaussian (also referred to as genuinely non-Gaussian) if  it is not contained in the convex hull generated by pure Gaussian states, i.e., if $\rho\notin\mathcal C_{\rm g}$ \cite{Filip2011,Geno2013,Hugh2014,Zhuang2018} with  
$\mathcal C_{\rm g}:=\overline{\rm conv}\, \mathcal{G}_d=\overline{\rm conv}\, \{\ket{ g}\bra{ g}: \ket{g}\in \mathcal{G}_d^{\rm pure} \}\supset \mathcal{G}_d$, with the inclusion being strict. }Physically, it means that it cannot be prepared, or approximated arbitrarily well, by classical randomization of Gaussian preparations. For an admissible orthogonal matrix $O\in \widetilde{G}_k$, identical copies of any pure Gaussian state are unchanged by the interferometer, $\Gamma_d(O)\ket{g}^{\otimes k}=\ket{g}^{\otimes k},$ and therefore attain the maximal value $f_k^{(d)}(\ket{g}\bra{g},O)\equiv 1$. In $\mathcal C_{\rm g}$, independently prepared copies can contain different components, so this perfect invariance need not persist. Nevertheless, purity guarantees a minimum weight in the sector left unchanged by $\Gamma_d(O)$. This sector contributes $+1$ to the real expectation, while the remaining eigenphases contribute cosines bounded below by $-1$. These observations yield a purity-dependent lower bound whose violation certifies quantum non-Gaussianity.

\begin{theorem}\label{th:genuine-nongau}
For any $\rho\in\mathcal C_{\rm g}$, $k\geq2$, and $O\in\widetilde{G}_k$,
\begin{equation}\label{eq:genuine-main}
 \gamma_1(\rho)=\operatorname{Re}\Big( {\rm Tr}\big(\rho^{\otimes k}\Gamma_d(O)\big) \Big) 
 \geq 2p_\rho^{\,k-1/2}-1,
\end{equation}
where $p_\rho$ is the global purity of the complete $d$-mode input. The criterion uses purity and the expectation $\gamma_1$, without requiring estimation of CM or 
 a finite-energy assumption. A violation certifies quantum non-Gaussianity.
\end{theorem}

A tighter bound can be obtained if we have more knowledge about $\rho$ as shown in SM (Sec. E.1) \cite{SupplementalMaterial}. More generally, for a fixed admissible $O$, higher powers, one
defines the real moments
$\gamma_s(\rho)=\operatorname{Re}\!\big(
 {\rm Tr}\big(\rho^{\otimes k}\Gamma_d(O)^s\big)\big)$, and the associated moment matrix $C_L=(\gamma_{|j-i|})_{i,j=0}^{L}$,
with $\gamma_0=1.$ All these moments are extracted from the same photon-counting record by classical postprocessing. The three-copy measurement protocol is illustrated in Fig.~\ref{fig:1}. Every $\rho\in\mathcal C_{\rm g}$ satisfies
 $R_L(\rho):=C_L-bJ_L\succeq0,
 \, L=0,1,\ldots$,
where $J_L=\boldsymbol1\boldsymbol1^T$ and $b=p_\rho^{k-1/2}$. These conditions form a nested hierarchy because each matrix is a principal submatrix of the next. The $L=1$ condition recovers Theorem~\ref{th:genuine-nongau}, while higher levels impose joint constraints on several moments and can detect states that pass the individual-power bounds. A negative eigenvalue at any level certifies quantum non-Gaussianity.

\textit{Conclusion}---We developed a method for detecting non-Gaussianity through the interference of identical copies. Admissible copy mixers yield exact Gaussian reference expectations independent of displacement and squeezing, providing witnesses for single-mode and multimode states. For single-mode states with finite mean photon number, these witnesses form a complete hierarchy over the copy number. We further derived purity-dependent bounds and joint-moment tests that certify quantum non-Gaussianity, excluding convex mixtures of Gaussian states.

\section*{END MATTER} Our approach complements the symmetry-based Gaussianity tests of Ref.~\cite{Girardi2025}, which emphasize rotation symmetries, robustness, and sample complexity. Beyond that framework, our formulation uses bosonic second quantization to obtain closed-form reference expectations for general mixed multimode Gaussian states, establishes a complete hierarchy over copy number, certifies quantum non-Gaussianity beyond convex mixtures, and provides explicit measurement protocols based on passive interferometry and photon-number-resolved detection.

This framework also suggests a connection to the resource theory of asymmetry
\cite{dai2023approximate,marvian2014extending}. For single-mode states, let $G\in\mathfrak{so}(k)$ satisfy $G\boldsymbol e_1=0$ and define $H_G=i\Gamma'(G)$. Gaussian product states are invariant under the corresponding copy rotations, so the quantum Fisher information $F_Q(\rho_g^{\otimes k},H_G)=0$. A positive $F_Q(\rho^{\otimes k},H_G)$ therefore certifies ordinary non-Gaussianity. It quantifies the sensitivity of $\rho$ to small copy rotations, while our witnesses compare expectations with Gaussian references. Relating these quantities and establishing their operational significance as measures of non-Gaussianity remain directions for future work.

\begin{acknowledgments}
The authors are grateful to Professor Shunlong Luo, and Yuwei Zhu for discussion. This work was supported by the National Natural Science Foundation of China, Grant Nos.~12401609 and 12501629, and the Youth Innovation Promotion
Association of CAS, Grant No.~2023004.
\end{acknowledgments}

\widetext
\section*{Supplemental Material: Detecting non-Gaussianity and quantum non-Gaussianity via multicopy interference}



This Supplemental Material provides complete derivations and supplementary explanations for several technical details in the main text. Appendix A fixes the representations and the interleaved quadrature ordering, as well as the definition of Gaussian states. Appendix B proves the single-mode symmetries and Gaussian reference, constructs infinitesimal generators, and derives the examples. Appendix C extends the construction to multimode states. Appendix D proves completeness under finite mean photon number, including the single-mode purity-calibrated family. Appendix E proves the quantum non-Gaussianity bound and the joint real-moment hierarchy. Appendix F gives the general $k$-copy passive measurement protocol and finite-sample analysis.

Main-text Theorems 1, 2, 3, 4, and 5 are restated and proved in B.3, B.4, C.2, D, and E.1, respectively. Supplementary statements and equations have the prefix S to distinguish them from the main text. We consistently write the admissible subgroup as $\widetilde G_k$.

\section*{Appendix A: Conventions and Preliminaries}

In this section, we clarify the notations and review some basic concepts, including the definitions and representations of relevant Lie groups and Lie algebras, as well as the definition of Gaussian states. 

\subsection*{A.1 Symplectic and orthogonal groups}

We use $\hbar=1$ and the same interleaved quadrature ordering as the main text,
\[
 \boldsymbol R=(q_1,p_1,\ldots,q_d,p_d)^T,\qquad q_r=\frac{a_r+a_r^\dagger}{\sqrt2},\quad p_r=\frac{a_r-a_r^\dagger}{i\sqrt2}.
\]
Thus $[R_j,R_\ell]=i\Omega_{j\ell}$, with
\begin{equation}\label{eq:symgroup}
 \operatorname{Sp}(2d,\mathbb R)
 =\{S\in GL(2d,\mathbb R):S^T\Omega S=\Omega\},
 \qquad
 \Omega=\bigoplus_{r=1}^{d}
 \begin{pmatrix}0&1\\-1&0\end{pmatrix}.
\end{equation}
The corresponding Lie algebra is
\begin{equation}\label{eq:symalgebra}
 \mathfrak{sp}(2d,\mathbb R)
 =\{Y\in M_{2d}(\mathbb R):Y^T\Omega+\Omega Y=0\}.
\end{equation}
Here $M_n(\mathbb R)$ denotes all real $n\times n$ matrices. Since $\Omega^2=-I$, the last condition is equivalently
$\Omega Y^T\Omega=Y$.

For copy mixing we use
\begin{equation}\label{eq:orthogonal group}
 O(k)=\{O\in GL(k,\mathbb R):O^TO=OO^T=I_k\},
\end{equation}
whose Lie algebra is
\begin{equation}\label{eq:ortho_algebra}
 \mathfrak{so}(k)=\{A\in M_k(\mathbb R):A^T=-A\}.
\end{equation}
Here $ GL(k,\mathbb R)$ denote the groups of invertible $k\times k$ real matrices.

The groups $O(k)$ and $SO(k)$ have the same Lie algebra, but $O(k)$ also contains the component with determinant $-1$ \cite{Hall2015}. Every $O\in SO(k)$ is $e^A$ for a real antisymmetric $A$. A matrix with determinant $-1$ cannot have this form, because $\det(e^A)=e^{\operatorname{Tr}A}=1$. The passive implementation below is defined for both components.

\subsection*{A.2 Bosonic second quantization and passive interferometers}

The occupation-number states $|n_1,\ldots,n_d\rangle$ form an orthonormal basis of the bosonic Fock space. Creation and annihilation operators satisfy
\begin{equation}\label{eq:commut_relat}
 [a_j,a_\ell^\dagger]=\delta_{j\ell}I,
 \qquad [a_j,a_\ell]=[a_j^\dagger,a_\ell^\dagger]=0.
\end{equation}
 We write
\[
    \boldsymbol{a}
    =
    (a_1,\ldots,a_d)^T,
    \qquad
    \boldsymbol{a}^\dagger
    =
    (a_1^\dagger,\ldots,a_d^\dagger)^T.
\]
For $A\in M_d(\mathbb C)$, define on the second quantization map \cite{woit2017quantum}
\begin{equation}\label{eq:algrep}
 \Gamma'(A)=\sum_{j,\ell=1}^{d}A_{j\ell}a_j^\dagger a_\ell.
\end{equation}
Using the canonical commutation relations \eqref{eq:commut_relat}, one obtains
\begin{equation*}
    \begin{aligned}
    [a_j^\dagger a_k,a_l^\dagger a_m]
    &= (a_j^\dagger a_k a_l^\dagger a_m-a_j^\dagger a_l^\dagger a_k a_m)+ (a_l^\dagger a_j^\dagger a_m a_k-a_l^\dagger a_m a_j^\dagger a_k)\\
    &=  \delta_{kl}a_j^\dagger a_m-\delta_{jm}a_l^\dagger a_k . 
    \end{aligned}
\end{equation*}
Therefore,
\begin{align*}
    [\Gamma'(A),\Gamma'(B)]
   & = \sum_{j,k,l,m} A_{jk} B_{lm} [a_j^\dagger a_k,\ a_l^\dagger a_m] \\
&= \sum_{j,k,l,m} A_{jk} B_{lm} \left( \delta_{kl} a_j^\dagger a_m - \delta_{jm} a_l^\dagger a_k \right) \\
&= \sum_{j,m} (AB)_{jm} a_j^\dagger a_m - \sum_{j,k} (BA)_{jk} a_j^\dagger a_k \\
&= \sum_{j,k} [A,B]_{jk} a_j^\dagger a_k \\
&= \Gamma'([A,B]), \qquad
    A,B\in\mathfrak{gl}(d,\mathbb{C}),
\end{align*}
which indicates that \(\Gamma'\) is a Lie algebra representation of \(\mathfrak{gl}(d,\mathbb{C})\) on the bosonic Fock space.

Then, for the matrix $e^A$, its bosonic second quantization representation is 
\begin{equation}\label{eq:alggro}
     \Gamma(e^A)=e^{ \Gamma'(A)}.
\end{equation}
The adjoint actions are
\begin{equation}\label{eq:metapro}
 \begin{aligned}
 \Gamma(e^A)\boldsymbol a\Gamma(e^A)^{-1}
   &=e^{-A}\boldsymbol a,\\
 \Gamma(e^A)\boldsymbol a^\dagger\Gamma(e^A)^{-1}
   &=e^{A^T}\boldsymbol a^\dagger,
 \end{aligned}
\end{equation}
where both displayed operator vectors are columns. These identities follow from $[\Gamma'(A),a_j]=-\sum_\ell A_{j\ell}a_\ell$ and $[\Gamma'(A),a_j^\dagger]=\sum_\ell A_{\ell j}a_\ell^\dagger$.

For a unitary matrix $V$, this convention gives $\Gamma(V)|0\rangle=|0\rangle$ and $\Gamma(V)^\dagger\boldsymbol a\Gamma(V)=V\boldsymbol a$. 
It fixes the otherwise undetermined overall phase, which is needed when evaluating an expectation of $\Gamma(V)$ itself. For a real orthogonal $O$, the convention in the main text is therefore
\[
 \Gamma(O)\boldsymbol a\Gamma(O)^{-1}=O^T\boldsymbol a.
\]
Orthogonality preserves $\sum_j a_j^\dagger a_j$, and reality ensures that the same coefficients mix $q_j$ and $p_j$.
Thus $\Gamma(O)$ implements passive copy interference, including reflections and permutations. 

\subsection*{A.3 Symplectic transformations in the complex basis}

Let $\boldsymbol\xi=(a_1,\ldots,a_d,a_1^\dagger,\ldots,a_d^\dagger)^T$. To relate this grouped complex vector to the interleaved real quadratures, let $P$ be the permutation matrix satisfying $P\boldsymbol R=(q_1,\ldots,q_d,p_1,\ldots,p_d)^T$. Then
\[
 \boldsymbol\xi=L_{(c)}\boldsymbol R,\qquad
 L_{(c)}=\frac1{\sqrt2}
 \begin{pmatrix}I_d&iI_d\\I_d&-iI_d\end{pmatrix}P,
 \qquad L_{(c)}^\dagger L_{(c)}=I_{2d}.
\]
A real symplectic matrix has the complex form $S^{(c)}=L_{(c)}SL_{(c)}^\dagger$, which satisfies
\[
 S^{(c)}KS^{(c)\dagger}=K,\qquad
 K=iL_{(c)}\Omega L_{(c)}^\dagger
 =\begin{pmatrix}I_d&0\\0&-I_d\end{pmatrix}.
\]
The corresponding Lie algebra in the complex basis consists of matrices of the form $  -iKH$ with \(H\) being a proper Hermitian matrix. Thus, whenever \(S^{(c)}\) can be written as a single exponential, 
\[
    S^{(c)}=e^{-iKH},
\]
one obtains a quadratic unitary operator
\begin{equation}\label{eq:symform-complex}
      \Lambda(S):=    \exp\!\Big(
        -\frac{i}{2}\boldsymbol{\xi}^\dagger H\boldsymbol{\xi}
    \Big).
\end{equation}
Using the canonical commutation relations in the compact form
\[
    [\xi_m,\xi_n^\dagger]=K_{mn},
\]
the Baker--Campbell--Hausdorff formula gives
\[
  \Lambda(S)^\dagger
\boldsymbol\xi
\Lambda(S)    =    e^{-iKH}\boldsymbol{\xi}    =    S^{(c)}\boldsymbol{\xi} .
\]
Hence quadratic bosonic unitaries realize the symplectic action as Bogoliubov transformations on the creation and annihilation operators.

The corresponding Gaussian unitary is defined, up to an overall phase, by the same Heisenberg action as in the main text \cite{Weed2012,adesso2014continuous},
\begin{equation}\label{eq:symrep}
 \Lambda(S)^\dagger\boldsymbol R\Lambda(S)=S\boldsymbol R.
\end{equation}
In the complex basis this is $\Lambda(S)^\dagger\boldsymbol\xi\Lambda(S)
=S^{(c)}\boldsymbol\xi$. For a symplectic matrix that is a single exponential, $S=e^Y$ with $Y\in\mathfrak{sp}(2d,\mathbb R)$, set $H=iK L_{(c)}YL_{(c)}^\dagger$.  The bosonic reality conditions give
\[
 H=H^\dagger
 =\begin{pmatrix}h&g\\\overline g&\overline h\end{pmatrix},
 \qquad h=h^\dagger,\quad g=g^T.
\]
The quadratic implementation is
\begin{equation}\label{eq:symform-real}
 \Lambda(e^Y)=
 \exp\!\left(-\frac{i}{2}\boldsymbol\xi^\dagger H\boldsymbol\xi\right)
 =\exp\!\left(\frac12\boldsymbol\xi^\dagger K
 L_{(c)}YL_{(c)}^\dagger\boldsymbol\xi\right).
\end{equation}
The canonical commutators $[\xi_j,\xi_\ell^\dagger]=K_{j\ell}$ imply its stated Bogoliubov action. Not every symplectic matrix is a single exponential. A general symplectic transformation is a finite product of such exponentials, and its unitary is the corresponding product. The commutation proofs below apply to each factor.

As an example, the single-mode squeezing operator with $\zeta=re^{i\theta}$ is
\[
 \Lambda(S_\zeta)
 =\exp\!\left(\frac12\zeta a^{\dagger2}
                  -\frac12\overline\zeta a^2\right),\qquad
 S_\zeta=
 \begin{pmatrix}
 \cosh r+\cos\theta\sinh r&\sin\theta\sinh r\\
 \sin\theta\sinh r&\cosh r-\cos\theta\sinh r
 \end{pmatrix}.
\]

\subsection*{A.4 Gaussian states and reference calibration}

A state is Gaussian when its characteristic function, or equivalently its Wigner function, is Gaussian. Define
$d_j=\operatorname{Tr}(\rho R_j)=\langle R_j\rangle_\rho$ and
\begin{equation}\label{eq:covariance-definition}
 \sigma_{j\ell}
 =\langle R_jR_l+R_lR_j\rangle_\rho-2\langle R_j\rangle_\rho \langle R_l\rangle_\rho .
\end{equation}
Thus the vacuum covariance is $I_{2d}$, as in the main text. Finite quadrature second moments are equivalent to finite total mean photon number because
\[
 \sum_{j=1}^{2d}\langle R_j^2\rangle_\rho =2\sum_{r=1}^{d}\langle a_r^\dagger a_r\rangle_\rho+d.
\]

The displacement operator is
\[
D(\boldsymbol\alpha)=\exp(\boldsymbol{\alpha}^{T}\boldsymbol{a}^\dagger-\overline{\boldsymbol{\alpha}}^{\,T}\boldsymbol{a})
 =\exp\!\left(\sum_{r=1}^{d}
 (\alpha_r a_r^\dagger-\overline\alpha_r a_r)\right),
 \qquad \boldsymbol\alpha\in\mathbb C^d.
\]
It shifts $a_r$ by $\alpha_r$ in the Heisenberg picture. Every Gaussian state is a displaced symplectic preparation of
independent thermal states, 
\begin{equation}\label{eq:multigauss}
 \rho_g=D(\boldsymbol\alpha)\Lambda(S)
 \rho_{\boldsymbol\lambda}
 \Lambda(S)^\dagger D(\boldsymbol\alpha)^\dagger.
\end{equation}
where 
\begin{equation}\label{eq:multithermo}
 \rho_{\boldsymbol\lambda}=
 \bigotimes_{r=1}^{d}\rho_{\lambda_r},\qquad
 \rho_{\lambda_r}=(1-\lambda_r)
 \sum_{n=0}^{\infty}\lambda_r^n|n\rangle\langle n|,
 \quad 0\leq\lambda_r<1,
\end{equation}
and $\boldsymbol{\lambda}=(\lambda_1,\cdots,\lambda_d).$

Here $\overline n_r=\lambda_r/(1-\lambda_r)$ is the mean photon number of the \(r\)-th thermal mode. In the interleaved quadrature ordering,
\begin{equation}\label{eq:thermal-covariance}
 \sigma_{\boldsymbol\lambda}
 =\bigoplus_{r=1}^{d}\nu_r I_2,
 \qquad \sigma_g=S\sigma_{\boldsymbol\lambda}S^T,
\end{equation}
where Williamson's symplectic eigenvalues and the thermal parameters are related by
\begin{equation}\label{eq:Williamson-thermal-parameters}
 \nu_r=\frac{1+\lambda_r}{1-\lambda_r},
 \qquad \lambda_r=\frac{\nu_r-1}{\nu_r+1}
 \quad(r=1,\ldots,d).
\end{equation}
Displacement sets the first moments and leaves the covariance unchanged. Gaussian unitaries preserve purity, giving
\[
 \operatorname{Tr}(\rho_g^2)
 =\prod_{r=1}^{d}\frac{1-\lambda_r}{1+\lambda_r}
 =\prod_{r=1}^{d}\nu_r^{-1}.
\]
For a single-mode Gaussian state, the unique thermal parameter can be determined directly from the purity by $\lambda=(1-p)/(1+p)$ with $p $ being purity.

\section*{Appendix B: Multi-Copy Symmetries and Gaussian Reference Values}
We prove the two symmetry lemmas and the single-mode reference formula used in the main text. These results identify the commuting subgroup within orthogonal copy mixers. We then construct its infinitesimal generators and evaluate the three-copy examples.

\subsection*{B.1 Proof of Lemma 1}
\begin{lemma}\label{appenlem:com_sym}
For every \(O\in O(k)\) and \(S\in Sp(2,\mathbb{R})\), there is
\begin{equation}
     \left[\Gamma(O),\Lambda(S)^{\otimes k}\right]=0.
\end{equation}
\end{lemma}
\begin{proof}
First suppose that $S=e^Y$. By Eq.~\eqref{eq:symform-complex}, for a symplectic matrix $S$, there exists a Hermitian matrix \(h\) (absorbing the factor $1/2$) such that the corresponding metaplectic operator can be written as
\begin{equation}
   \Lambda(S)=\exp\!\left(-i\xi^\dagger h\xi\right),
\end{equation}
where $ \xi=(a,a^\dagger)^T .$

For the \(k\)-copy system, define
\[
    \xi_\alpha
    =\begin{pmatrix}
        a_\alpha\\
        a_\alpha^\dagger
    \end{pmatrix},
    \qquad
    \alpha=1,\ldots,k .
\]
Since the same symplectic transformation acts on each copy, the quadratic generator of \(\Lambda(S)^{\otimes k}\) is
\begin{equation}
    G'    :=    \sum_{\alpha=1}^k    \xi_\alpha^\dagger h_\alpha \xi_\alpha ,
\end{equation}
where \(h_\alpha=h\) for all \(\alpha\). Hence
\begin{equation}
    \Lambda(S)^{\otimes k}    =    e^{-iG'}.
\end{equation}

Equivalently, if we introduce the block vector $ \Xi=    \bigoplus_{\alpha=1}^k \xi_\alpha$ and the block-diagonal matrix
$H=\bigoplus_{\alpha=1}^k h_\alpha,$ then the generator can be written compactly as
\begin{equation}
    G'=\Xi^\dagger H\Xi .
\end{equation}
By Eq.~\eqref{eq:metapro}, the second-quantized orthogonal transformation \(\Gamma(O)\) mixes only the copy indices:
\begin{equation}
\begin{aligned}
    \Gamma(O)\boldsymbol{a}\Gamma(O)^{-1}
    &=
    \overline{O^T}\boldsymbol{a}
    =
    O^T\boldsymbol{a},\\
    \Gamma(O)\boldsymbol{a}^\dagger\Gamma(O)^{-1}
    &=
    O^T\boldsymbol{a}^\dagger .
\end{aligned}
\end{equation}
Here we used the fact that \(O\in O(k)\) is real. Therefore, for each copy \(\alpha\),
\begin{equation}
\begin{aligned}
    \Gamma(O)\xi_\alpha\Gamma(O)^{-1}
    &=
    \Gamma(O)
    \begin{pmatrix}
        a_\alpha\\
        a_\alpha^\dagger
    \end{pmatrix}
    \Gamma(O)^{-1}  \\
    &=
    \begin{pmatrix}
        \sum_{\beta=1}^k O_{\beta\alpha}a_\beta\\
        \sum_{\beta=1}^k O_{\beta\alpha}a_\beta^\dagger
    \end{pmatrix}  \\
    &=
    \sum_{\beta=1}^k O_{\beta\alpha}\xi_\beta .
\end{aligned}
\end{equation}

We now compute the conjugation of the quadratic generator. Since \(h_\alpha=h\) for all copies, we have 
\begin{equation}
\begin{aligned}
\Gamma(O) G' \Gamma(O)^{-1}&=\sum_\alpha \Gamma(O)  \xi_\alpha^\dagger h_\alpha\xi_\alpha \Gamma(O)^{-1}\\
&=\sum_{\alpha=1}^k \sum_{i,j=1,2}\Gamma(O)\xi_{\alpha,i}^{\dagger}h_{\alpha,ij}\xi_{\alpha,j} \Gamma(O)^{-1}\\
&=\sum_{\alpha=1}^k \sum_{i,j=1,2}h_{\alpha,ij}\Gamma(O) \xi_{\alpha,i}^{\dagger}\Gamma(O)^{-1}\Gamma(O)\xi_{\alpha,j} \Gamma(O)^{-1}\\
&=\sum_{\alpha=1}^k \Gamma(O)  \xi_\alpha^\dagger  \Gamma(O)^{-1} h_\alpha \Gamma(O)\xi_\alpha \Gamma(O)^{-1}\\
&=\sum_{\alpha=1}^k \Gamma(O)  \xi_\alpha^\dagger  \Gamma(O)^{-1} h \Gamma(O)\xi_\alpha \Gamma(O)^{-1}\\
&=\sum_{\alpha=1}^k \Big(\sum_{\beta=1}^{k}O_{\beta\alpha}\xi_\beta^\dagger\Big) h \Big(\sum_{\gamma=1}^k O_{\gamma\alpha}\xi_\gamma\Big) \\
&=\sum_{\alpha,\beta,\gamma=1}^k O_{\beta\alpha}O_{\gamma\alpha}\xi_\beta^\dagger h \xi_\gamma.
\end{aligned}  
\end{equation}
Since \(O\) is orthogonal, its rows are orthonormal. Hence,
\begin{equation}
    \sum_{\alpha=1}^k
    O_{\beta\alpha}O_{\gamma\alpha}
    =
    \delta_{\beta\gamma}.
\end{equation}
Therefore,
\begin{equation}
\begin{aligned}
    \Gamma(O)G'\Gamma(O)^{-1}
    &=    \sum_{\beta,\gamma=1}^k
    \delta_{\beta\gamma}
    \xi_\beta^\dagger h\xi_\gamma  \\
    &=    \sum_{\beta=1}^k
    \xi_\beta^\dagger h\xi_\beta
    =    G' .
\end{aligned}
\end{equation}
Thus, the quadratic generator \(G'\) is invariant under conjugation by \(\Gamma(O)\). Exponentiating this identity gives
\begin{equation}
    \Gamma(O)\Lambda(S)^{\otimes k}\Gamma(O)^{-1}
    =
    \Lambda(S)^{\otimes k}.
\end{equation}
Equivalently, for every single-exponential factor,
\begin{equation}
    \Gamma(O)\Lambda(S)^{\otimes k}
    =
    \Lambda(S)^{\otimes k}\Gamma(O).
\end{equation}
For a general $S\in\operatorname{Sp}(2,\mathbb R)$, write $S=e^{Y_1}\cdots e^{Y_m}$. The corresponding unitary is, up to an irrelevant overall phase, $\Lambda(e^{Y_1})\cdots\Lambda(e^{Y_m})$. The established commutation identity holds for every factor and hence for their product. This proves the assertion for all $S$.
\end{proof}

As a remark, this commutation relation is a concrete realization of the orthogonal and symplectic duality appearing in the oscillator representation, commonly described in terms of Howe duality \cite{goodman2009symmetry,goodman2004multiplicity}. The direct calculation above is sufficient for the results used here.

\subsection*{B.2 Proof of Lemma 2}
Define the collective direction and its annihilation operator by
\[
 \boldsymbol e_1=k^{-1/2}(1,\ldots,1)^T,
 \qquad c_1=\boldsymbol e_1^T\boldsymbol a.
\]
Applying the  $k$-copy identical displacement operator, a common drive occupies only this mode, and 
\begin{equation}
 D(\alpha)^{\otimes k}= \exp\left( \alpha\sum_{i=1}^k a_i^\dagger -\overline{\alpha}\sum_{i=1}^k a_i\right)
 =\exp\!\left(\sqrt{k} (\alpha c_1^\dagger-\overline\alpha c_1)\right).
\end{equation}
The remaining $k-1$ relative modes receive no common displacement.

\begin{lemma}\label{appenlem:dis_com}
Suppose \(O\in O(k)\). Then
\begin{equation}
 [\Gamma(O),D(\alpha)^{\otimes k}]=0,
 \qquad \forall\alpha\in\mathbb C,
\end{equation}
if and only if $O\boldsymbol e_1=\boldsymbol e_1$.
\end{lemma}
\begin{proof}
Rewrite $k$-copy displacement operator as
$$
 D_a(\alpha)^{\otimes k}
 =\exp\Big(\alpha\sum_i a_i^\dagger-\bar\alpha\sum_i a_i\Big)
 =D_c(\sqrt k\,\alpha),
$$
so we only need to show
\begin{equation}\label{eq:C}
 [\Gamma(O),D_c(z)]=0 \quad \forall z\in\mathbb C
 \qquad\Longleftrightarrow\qquad O\boldsymbol e_1=\boldsymbol e_1 .
\end{equation}

\emph{Necessity}: If $[\Gamma(O),D_c(z)]=0$ for all $z$, then $\Gamma(O)D_c(z)\Gamma(O)^{-1}=D_c(z)$ for all $z$. Differentiating this identity in $z$ and in $\bar z$ at $z=0$ gives
$$
 \Gamma(O)c_1^\dagger\Gamma(O)^{-1}=c_1^\dagger,
 \qquad
 \Gamma(O)c_1\Gamma(O)^{-1}=c_1 .
$$
Combining Eq. \eqref{eq:metapro},  yields
$$
 (O\boldsymbol e_1)^T\boldsymbol a=\boldsymbol e_1^T\boldsymbol a .
$$
By linear independence of $a_1,\dots,a_k$ the coefficients match componentwise, hence $O\boldsymbol e_1=\boldsymbol e_1$. This establishes \eqref{eq:C}.

\emph{Sufficiency}: If $O\boldsymbol e_1=\boldsymbol e_1$, then $\Gamma(O)c_1^{\dagger}\Gamma(O)^{-1}=c_1^{\dagger}$, so $\Gamma(O)$ fixes the generator of $D_c(z)=\exp(zc_1^\dagger-\bar zc_1)$, hence commutes with $D_c(z)$ for every $z$.
\end{proof}

\subsection*{B.3 Proof of Theorem 1}
\begin{theorem}\label{appenth:main}
For $\alpha\ne0$ and $S\in\operatorname{Sp}(2,\mathbb R)$, an orthogonal copy mixer $\Gamma(O)$ commutes with $(D(\alpha)\Lambda(S))^{\otimes k}$ if and only if 
\begin{equation}\label{eq:thcon}
 O\boldsymbol e_1=\boldsymbol e_1,
 \qquad O\in O(k).
\end{equation}
The admissible orthogonal operators form the subgroup $\widetilde G_k\cong O(k-1)$.
\end{theorem}
\begin{proof}
First note that tensor products distribute over operator products, so
\begin{equation}
    (D(\alpha)\Lambda(S))^{\otimes k}
    =D(\alpha)^{\otimes k}\Lambda(S)^{\otimes k}.
\end{equation}
By Lemma~\ref{appenlem:com_sym}, commutation with $(D(\alpha)\Lambda(S))^{\otimes k}$ is equivalent to commutation with $D(\alpha)^{\otimes k}$. Combine Lemmas \ref{appenlem:com_sym} and Lemma \ref{appenlem:dis_com}, and we can obtain that
\begin{equation}
    [\Gamma(O),(D_\alpha \Lambda(S))^{\otimes k}]=0
\end{equation}
holds if and only if \(O\in O(k)\) and \(O\boldsymbol{e}_1=\boldsymbol{e}_1\).

It remains to prove that the set of all such operators forms a subgroup isomorphic to \(O(k-1)\). Extend \(\boldsymbol{e}_1\) to an orthonormal basis $\{\boldsymbol{e}_1,\ldots,\boldsymbol{e}_k\}$ of \(\mathbb{R}^k\), and let \(C=(\boldsymbol{e}_1,\ldots,\boldsymbol{e}_k)\) be the corresponding orthogonal matrix.

If \(O\in O(k)\) satisfies \(O\boldsymbol{e}_1=\boldsymbol{e}_1\), then orthogonality also gives \(O^T\boldsymbol{e}_1=O^{-1}\boldsymbol{e}_1=\boldsymbol{e}_1\). Therefore, in the basis defined by \(C\), the matrix \(O\) has the block diagonal form
\begin{equation}
    C^T O C
    =
    \begin{pmatrix}
        1&0\\
        0&\widetilde{O}
    \end{pmatrix},
\end{equation} 
where \(\widetilde{O}\in O(k-1)\). Conversely, every matrix of this block form defines an orthogonal matrix fixing \(\boldsymbol{e}_1\). Thus,
\begin{equation}
    \{O\in O(k):O\boldsymbol{e}_1=\boldsymbol{e}_1\}
    \cong O(k-1).
\end{equation}
\end{proof}

\subsection*{B.4 Proof of Theorem 2}
\begin{theorem}[Reference value]\label{appenth:refval}
For $0\leq\lambda<1$ and $O\in\widetilde G_k$,
\begin{equation}\label{eq:reference}
 W_k(\lambda,O)=
 \frac{(1-\lambda)^k}{\det(I_k-\lambda O)},
\end{equation}
where $I_k$ is the identity operator.
\end{theorem}
\begin{proof}
Put $V_g=D(\alpha)\Lambda(S)$, so that $\rho_g=V_g\rho_\lambda V_g^\dagger$. Using cyclicity of the trace and the two symmetry lemmas gives
\begin{equation}
\begin{aligned}
    \operatorname{Tr}\!\left(\rho_g^{\otimes k}\Gamma(O)\right)
    &=\operatorname{Tr}\!\left(\rho_\lambda^{\otimes k}
    V_g^{\dagger\otimes k}\Gamma(O)V_g^{\otimes k}\right)\\
    &=    \operatorname{Tr}\!\left(   \rho_{\lambda}^{\otimes k}\Gamma(O)    \right).
\end{aligned}
\end{equation}
Here the equalities are given by Lemmas \ref{appenlem:dis_com} and \ref{appenlem:com_sym}. Thus, for Gaussian states, this expectation value depends only on the thermal parameter \(\lambda\). 

For \(\lambda=0\), \(\rho_0=\ket{0}\bra{0}\). Since every passive copy-mixing unitary leaves the multimode vacuum invariant,
\begin{equation}
    W_k(0,O)
    =\bra{0}^{\otimes k}\Gamma(O)\ket{0}^{\otimes k}
    =1,
\end{equation}
in agreement with Eq.~\eqref{eq:reference}.

It remains to consider \(0<\lambda<1\). With $\hat{n}=a^\dagger a$, the thermal state can be written as
\begin{equation}
    \rho_{\lambda}
    =
    \frac{\lambda^{a^\dagger a}}
    {\operatorname{Tr}(\lambda^{a^\dagger a})}
    =
    \frac{e^{\beta \hat{n}}}
    {\operatorname{Tr}(e^{\beta \hat{n}})},
\end{equation}
where $\beta=\ln\lambda<0$. Thus,
\begin{equation}
    \rho_{\lambda}^{\otimes k}
    =
    \frac{e^{\beta\sum_{i=1}^k \hat{n}_i}}
    {\operatorname{Tr}\left(e^{\beta\sum_{i=1}^k \hat{n}_i}\right)}.
\end{equation}
The desired expectation value becomes
\begin{equation}
\begin{aligned}
    W_k(\lambda,O)
    &=    \operatorname{Tr}\!\left(
    \rho_{\lambda}^{\otimes k}\Gamma(O)
    \right)                                                     \\
    &=    \frac{
    \operatorname{Tr}\!\left(
    e^{\beta\sum_{i=1}^k \hat{n}_i}\Gamma(O)
    \right)}
    {\operatorname{Tr}\!\left(
    e^{\beta\sum_{i=1}^k \hat{n}_i}
    \right)} .
\end{aligned}
\end{equation}
A direct calculation gives $  \operatorname{Tr}\!\left(    e^{\beta\sum_{i=1}^k \hat{n}_i}    \right)    =    (1-\lambda)^{-k},$ and it remains to evaluate the numerator. Using the second-quantized representation, 
\begin{equation}
    \begin{aligned}
        \operatorname{Tr}(  e^{\beta \sum_{i=1}^k \hat{n}_i  }\Gamma(O)) &=\operatorname{Tr}[\Gamma(e^{\beta I_k})\Gamma(O)]\\
        &=\operatorname{Tr}(\Gamma(e^{\beta}O))\\
        &=\operatorname{Tr}(\Gamma(V)\Gamma(e^{\beta}D)\Gamma(V^{-1}))\\
        &=\operatorname{Tr}(\Gamma(\lambda D)),
    \end{aligned}
\end{equation}
where we use the spectral decomposition $O=VDV^{-1}$ with $V$ being orthogonal and $D={\rm diag}(d_1,\cdots, d_k)$ being a diagonal matrix and each $d_i=e^{i\phi_i}$.

For a Fock state $\ket{n_1,\cdots,n_k}$, 
\begin{equation}
    \Gamma(\lambda D) \ket{n_1,\cdots,n_k}=\exp\Big( \sum_{i=1}^k \ln (\lambda d_i) N_i \Big)\ket{n_1,\cdots,n_k}=\prod_{i=1}^k(\lambda d_i)^{n_i}\ket{n_1,\cdots,n_k}.
\end{equation}
Consequently,
\begin{equation}
\begin{aligned}
    \operatorname{Tr}\!\left(\Gamma(\lambda D)\right)
    &=    \sum_{n_1,\ldots,n_k=0}^{\infty}
    \prod_{i=1}^k(\lambda d_i)^{n_i}                         \\
    &=    \prod_{i=1}^k
    \left(    \sum_{n_i=0}^{\infty}
    (\lambda d_i)^{n_i}    \right)                                                   \\
    &=    \prod_{i=1}^k
    \frac{1}{1-\lambda d_i}                                  \\
    &=    \frac{1}{\det(I_k-\lambda D)}=   \frac{1}{\det(I_k-\lambda O)}.
\end{aligned}
\end{equation}
As a result, the expectation value takes the form of Eq. \eqref{eq:reference}.

\end{proof}

\subsection*{B.5 Generators and single-mode examples}

We construct infinitesimal generators for the continuous copy rotations and derive the examples in the main text. Let \(\widetilde{G}_k\) denote the subgroup appearing in Theorem~\ref{appenth:main}, i.e.,
\begin{equation}
    \widetilde{G}_k:=\{O\in O(k):O\boldsymbol{e}_1=\boldsymbol{e}_1\},
\end{equation}
and let \(\widetilde{\mathfrak{g}}_k\) be its Lie algebra. As a special case, \(k=2\), $\widetilde{G}_2$ is isomorphic to \(O(1)\), which is discrete. Hence it has no nontrivial Lie algebra generator. However, there is one nontrivial group element, which can be chosen as the swap matrix
\begin{equation}
    \mathsf{SWAP} =\begin{pmatrix}
        0&1\\
        1&0
    \end{pmatrix}.
\end{equation}
One directly checks that  
\begin{equation}
    \mathsf{SWAP}(1,1)^T=(1,1)^T.
\end{equation}
The corresponding operator \(\Gamma(\mathsf{SWAP})\) is precisely the swap operator on two copies. Consequently, 
\begin{equation}
\begin{aligned}
    \operatorname{Tr}\!\left(\rho_g^{\otimes 2}\Gamma(\mathsf{SWAP})\right)
    &=
    \operatorname{Tr}(\rho_g ^2)                                      \\
    &=
    \operatorname{Tr}(\rho_{\lambda}^2)                             \\
    &=
    \frac{1-\lambda}{1+\lambda}.
\end{aligned}
\end{equation}
For any single-mode input, let $p_\rho=\operatorname{Tr}\rho^2$.
The purity-matched thermal reference parameter is
\begin{equation}\label{eq:pur_ther}
 \lambda=\frac{1-p_{\rho}}{1+p_{\rho}},\qquad
 \operatorname{Tr}\rho_\lambda^2=p_{\rho}.
\end{equation}
This calibrates a reference, without assigning a temperature to the unknown state. The only two-copy tests are identity and SWAP, whose expectations $1$ and $p$ agree with this reference for every state. Nontrivial single-mode certification therefore starts at three copies.

\begin{lemma}\label{applem:generator}
For \(k\geq3\) and \(1\leq i<j\leq k-1\), define
\begin{equation}\label{eq:Gij}
    G_{ij}
    =    E_{ij}-E_{ji}    +    E_{ki}-E_{ik}    +    E_{jk}-E_{kj},
\end{equation}
where \(E_{ab}\) is the matrix with a \(1\) in position \((a,b)\) and zeros elsewhere. Then the set
\begin{equation}
    \{G_{ij}:1\leq i<j\leq k-1\}
\end{equation}
forms a basis of \(\widetilde{\mathfrak{g}}_k\).
\end{lemma}
\begin{proof}
   Since \(E_{ab}^T=E_{ba}\), we have 
   \begin{equation}
        G^T_{ij}=E_{ji}-E_{ij}+E_{ik}-E_{ki}+E_{kj}-E_{jk}=-G_{ij},
    \end{equation}
which implies that \(G_{ij}\in\mathfrak{so}(k)\) and consequently $e^{\theta G_{ij} }$ is orthogonal. 

Now we check the action on \(\boldsymbol{e}_1\). Since $\boldsymbol{e}_1=\frac{1}{\sqrt{k}}(1,\ldots,1)^T,$
\begin{equation}
        G_{ij}\boldsymbol{e}_1 =(E_{ij}-E_{ik})\boldsymbol{e}_1+(E_{jk}-E_{ji})\boldsymbol{e}_1+(E_{ki}-E_{kj})\boldsymbol{e}_1=0.
    \end{equation}
    Therefore,
\begin{equation}
       e^{\theta G_{ij} }\boldsymbol{e}_1=I \boldsymbol{e}_1+\sum_{n=1}^{\infty} \frac{(\theta G_{ij})^n}{n!}\boldsymbol{e}_1=\boldsymbol{e}_1.
\end{equation}
which implies that each \(G_{ij}\) lies in the Lie algebra of \( \widetilde{G}_k\).

It remains to show that these matrices form a basis. The number of matrices is
\begin{equation}
    \binom{k-1}{2}
    =
    \frac{(k-1)(k-2)}{2},
\end{equation}
which is the dimension of \(\mathfrak{so}(k-1)\). Since \( \widetilde{G}_k\cong O(k-1)\), this is also the dimension of the Lie algebra \(\widetilde{\mathfrak{g}}_k.\) Moreover, the matrices \(G_{ij}\) are linearly independent. Indeed, for each pair $(i,j)$ with \(1\leq i<j\leq k-1\), the $(i,j)$-entry of $G_{ij}$ is $1$ while for any other $G_{i'j'}$ with $(i,j)\neq (i',j')$, the $(i,j)$-entry is $0$. Each $G_{ij}$ possesses a unique nonzero matrix element not shared by any other matrix in the set, which guarantees linear independence. Thus, the set \(\{G_{ij}\}\) is a basis of the Lie algebra of \( \widetilde{G}_k\).
\end{proof}
These generators span the Lie algebra of the connected component and they do not replace the determinant-$-1$ component of $\widetilde G_k$. The exact hierarchy in Appendix D tests all admissible orthogonal matrices, including reflections and permutations.

We now specialize to the case $k=3$. The group $\widetilde G_3\cong O(2)$ acts on a two-dimensional space of relative modes. Its connected component $SO(2)$ consists of rotations parameterized by one angle and has one independent generator. The sign convention below is the one used in the main text with its generator is $-G_{12}$ from Eq.~\eqref{eq:Gij}. Choose the generator 
\begin{equation}\label{eq:G}
    G = \begin{pmatrix}
0 & -1 & 1 \\
1 & 0 & -1 \\
-1 & 1 & 0
\end{pmatrix},
\end{equation}
and $O = e^{\theta G}\in \widetilde{G}_3$ for arbitrary $\theta\in \mathbb{R}$. The infinitesimal second-quantized generator associated with Eq.~\eqref{eq:G} is 
\begin{equation}\label{eq:A}
    \Gamma'(G)=a_2^\dagger a_1-a_1^\dagger a_2
    +a_1^\dagger a_3-a_3^\dagger a_1
    +a_3^\dagger a_2-a_2^\dagger a_3.
\end{equation}
Accordingly,
\begin{equation}\label{eq:ethetaG}
    \Gamma(e^{\theta G})=e^{\theta \Gamma'(G)}.
\end{equation}

The eigenvalues of $G$ are $0, \pm i\sqrt{3}$, so the eigenvalues of $O$ are $1, e^{i\sqrt{3}\theta}, e^{-i\sqrt{3}\theta}$. It follows that 
\[
\det(I_3 - \lambda O) = (1-\lambda)(1 - 2\lambda \cos(\sqrt{3}\theta) + \lambda^2).
\]
Therefore, the Gaussian reference value is
\begin{equation}\label{eq:k3example}
    W_3(\lambda,\theta) := W_3(\lambda, e^{\theta G})=\frac{(1-\lambda)^2}{1 - 2\lambda\cos(\sqrt{3}\theta) + \lambda^2}.
\end{equation}

\textbf{Example 1: Fock states.} 
Consider the Fock state \(\ket{n}\), with \(n\geq1\). To prove that \(\ket{n}\) is non-Gaussian, take the witness operator as Eq. \eqref{eq:ethetaG} and it is enough to show that the expectation value
\begin{equation}
    f_3(\ket{n}\bra{n},\theta):=\bra{n}^{\otimes 3}e^{\theta \Gamma'(G)}\ket{n}^{\otimes 3}
\end{equation}
is not identically equal to \(1\), which is the value predicted for any pure Gaussian state. We give two methods to establish this.

\textit{Method 1: Expansion around \(\theta=0\).}
The derivatives of $f_3(\ket{n}\bra{n},\theta)$ are
\begin{equation}
    \frac{{\rm d}^m}{{\rm d}\theta^m}f_3(\ket{n}\bra{n},\theta)
    =    \bra{n}^{\otimes 3}\Gamma'(G)^m e^{\theta \Gamma'(G)}\ket{n}^{\otimes 3}.
\end{equation}
At \(\theta=0\), the first derivative is
\begin{equation}
    f_3'(\ket{n}\bra{n},0)
    =    \bra{n}^{\otimes 3}\Gamma'(G)\ket{n}^{\otimes 3}
    =    0.
\end{equation}
For the second derivative, since \(\Gamma'(G)^\dagger=-\Gamma'(G)\), we obtain
\begin{equation}
\begin{aligned}
    f_3''(\ket{n}\bra{n},0)
    &=    \bra{n}^{\otimes 3}\Gamma'(G)^2\ket{n}^{\otimes 3}        \\
    &=    -\bra{n}^{\otimes 3}\Gamma'(G)^\dagger \Gamma'(G)\ket{n} ^{\otimes 3}\\
    &=    -6n(n+1).
\end{aligned}
\end{equation}
Therefore, \(f_3''(\ket{n}\bra{n},0)\neq0\) for \(n\geq1\), so \(f_3(\ket{n}\bra{n},\theta)\) is not constant. This indicates that \(\ket{n}\) is not a pure Gaussian state.

\textit{Method 2: Direct evaluation of the matrix element. }
As an alternative, fix
\begin{equation}
    \theta_\star=\frac{\pi}{2\sqrt{3}},
    \qquad    O_\star:=O(\theta_\star)
\end{equation}
in Eq. \eqref{eq:ethetaG}. We show directly that $\bra{n}^{\otimes 3}\Gamma(O_\star)\ket{n}^{\otimes 3}\neq1$. Consider the specific transition amplitude 
\begin{equation}
    C_n    =    \bra{3n,0,0}\Gamma(O_\star)\ket{n,n,n}.
\end{equation}
Using the rule
\begin{equation}
    \Gamma(O_\star)a_i^\dagger\Gamma(O_\star)^{-1}
    =    \sum_{j=1}^3 (O_\star)_{ji}a_j^\dagger,
\end{equation}
a direct calculation yields 
\begin{equation}
\begin{aligned}
    C_n
    &=
    \frac{\sqrt{(3n)!}}{(n!)^{3/2}}
    (O_{11}O_{12}O_{13})^n  \\
    &=
    \frac{\sqrt{(3n)!}}{(n!)^{3/2}}
    \left(-\frac{2}{27}\right)^n .
\end{aligned}
\end{equation}
Thus for every \(n\geq1\), \(C_n\neq0\). And \(\Gamma(O_\star)\ket{n}^{\otimes 3}\) has a nonzero component orthogonal to\(\ket{n}^{\otimes 3}\). Since \(\Gamma(O_\star)\) is unitary, this implies
\begin{equation}
    \left|     \bra{n}^{\otimes 3}\Gamma(O_\star)\ket{n}^{\otimes 3}   \right|<1.
\end{equation}
In particular, its expectation value cannot equal the pure-Gaussian reference value \(1\). 

As a complementary remark, we now derive an explicit formula for the matrix element \(\bra{n}^{\otimes 3}\Gamma(O)\ket{n}^{\otimes 3}\) for an arbitrary orthogonal operator \(O\in\widetilde{G}_3\). The Fock state can be generated from the vacuum state by
\begin{equation}
    \ket{n}^{\otimes 3}   =    \prod_{i=1}^3    \frac{(a_i^\dagger)^n}{\sqrt{n!}}    \ket{0}.
\end{equation}
Then
\begin{equation}
\begin{aligned}
    \bra{n}^{\otimes 3}\Gamma(O)\ket{n}^{\otimes 3}
    =    \frac{1}{(n!)^3}    \bra{0}    \prod_{i=1}^3 a_i^n    \prod_{j=1}^3    \Big(        \sum_{r=1}^3 O_{rj}a_r^\dagger    \Big)^n
    \ket{0}.
\end{aligned}
\end{equation}
To simplify the expansion, introduce the vector 
\begin{equation}
    J=
    (\overbrace{1,\ldots,1}^{n},
    \overbrace{2,\ldots,2}^{n},
    \overbrace{3,\ldots,3}^{n})
\end{equation}
of length \(N=3n\). Expanding the product gives
\begin{equation}
\begin{aligned}
    \prod_{j=1}^3
    \Big(
        \sum_{r=1}^3 O_{rj}a_r^\dagger
    \Big)^n
    =
    \sum_{r_1,\ldots,r_N\in\{1,2,3\}}
    \Big(
        \prod_{\beta=1}^N O_{r_\beta J_\beta}
    \Big)
    a_{r_1}^\dagger\cdots a_{r_N}^\dagger.
\end{aligned}
\end{equation}
Substituting this into the expectation value yields 
\begin{equation}
\begin{aligned}
    \bra{n}^{\otimes 3}\Gamma(O)\ket{n}^{\otimes 3}
    =    \frac{1}{(n!)^3}
    \sum_{r_1,\ldots,r_N\in\{1,2,3\}}
    \Big(
        \prod_{\beta=1}^N O_{r_\beta J_\beta}
    \Big)
    \bra{0}
    a_1^n a_2^n a_3^n
    a_{r_1}^\dagger\cdots a_{r_N}^\dagger
    \ket{0}.
\end{aligned}
\end{equation}
The vacuum expectation value is nonzero only if the sequence \((r_1,\ldots,r_N)\) contains exactly \(n\) entries equal to \(1\), \(n\) entries equal to \(2\), and \(n\) entries equal to \(3\). In that case,
\begin{equation}
    \bra{0}    a_1^n a_2^n a_3^n    a_{r_1}^\dagger\cdots a_{r_N}^\dagger
    \ket{0}    =    (n!)^3.
\end{equation}
Hence,
\begin{equation}
\begin{aligned}
    \bra{n}^{\otimes 3}\Gamma(O)\ket{n}^{\otimes 3}
    &=    \sum_{\substack{r_1,\ldots,r_N\in\{1,2,3\}\\
    n_1=n_2=n_3=n}}
    \prod_{\beta=1}^N O_{r_\beta J_\beta}        \\
    &=    \frac{\operatorname{Per}(O[nnn|nnn])}{(n!)^3}.
\end{aligned}
\end{equation}
Here \(n_i=\#\{\beta:r_\beta=i\}\), and \(\operatorname{Per}(O[nnn|nnn])\) is the permanent of the matrix
\begin{equation}
    O[nnn|nnn]
    =
    \begin{pmatrix}
        O_{11}J_n& O_{12}J_n& O_{13}J_n\\
        O_{21}J_n& O_{22}J_n& O_{23}J_n\\
        O_{31}J_n& O_{32}J_n& O_{33}J_n
    \end{pmatrix},
\end{equation}
where \(J_n\) is the \(n\times n\) all-ones matrix.

\textbf{Example 2: Even cat states.} 
The even cat state is the superposition of two coherent states,
\begin{equation}
    \ket{C_+}
    =\frac1{\mathcal{N}_+}
    \left(
        \ket{\alpha}+\ket{-\alpha}
    \right),
\end{equation}
where the normalization factor is
\begin{equation}
    \mathcal{N}_+
    =\sqrt{2(1+e^{-2|\alpha|^2})}.
\end{equation}
Again take \(k=3\), and use the orthogonal operator of the form in Eq.~\eqref{eq:ethetaG}. For any pure Gaussian state, the corresponding expectation value equals \(1\). Therefore, it is enough to show that the function
\[
    \bra{C_+}^{\otimes 3}
    e^{\theta \Gamma'(G)}
    \ket{C_+}^{\otimes 3}
\]
is not constant in \(\theta\). The first derivative at \(\theta=0\) vanishes by symmetry. We therefore calculate the second derivative.

Introduce the odd cat state
\begin{equation}
    \ket{C_-}
    =
    \mathcal{N}_-
    \left(
        \ket{\alpha}-\ket{-\alpha}
    \right),
    \qquad
    \mathcal{N}_-
    =
    \frac{1}{\sqrt{2(1-e^{-2|\alpha|^2})}}.
\end{equation}
A straightforward calculation gives 
\begin{equation}
\begin{aligned}
    a\ket{C_+}
   =
   \frac{\alpha}{\mathcal{N}_+\mathcal{N}_-}
    \ket{C_-}                             =
    \alpha\sqrt{\tanh |\alpha|^2}\ket{C_-}.
\end{aligned}
\end{equation}
Define the unnormalized photon-added even cat state by
\begin{equation}
    \ket{\chi}=a^\dagger\ket{C_+}.
\end{equation}
Using the expression for \(\Gamma'(G)\), we obtain
\begin{equation}
\begin{aligned}
    \Gamma'(G)\ket{C_+}^{\otimes 3}
    &=
    \alpha\sqrt{\tanh|\alpha|^2}
    \Big(
    \ket{C_-}_1\ket{\chi}_2\ket{C_+}_3
    -
    \ket{\chi}_1\ket{C_-}_2\ket{C_+}_3     \\
    &\quad
    +
    \ket{\chi}_1\ket{C_+}_2\ket{C_-}_3
    -
    \ket{C_-}_1\ket{C_+}_2\ket{\chi}_3     \\
    &\quad
    +
    \ket{C_+}_1\ket{C_-}_2\ket{\chi}_3
    -
    \ket{C_+}_1\ket{\chi}_2\ket{C_-}_3
    \Big).
\end{aligned}
\end{equation}
Since \(\Gamma'(G)^\dagger=-\Gamma'(G)\), the second derivative at \(\theta=0\) is
\begin{equation}
\begin{aligned}
    \frac{{\rm d}^2}{{\rm d}\theta^2}    \bra{C_+}^{\otimes 3}    e^{\theta \Gamma'(G)}    \ket{C_+}^{\otimes 3}    \bigg|_{\theta=0}
    &=-    \bra{C_+}^{\otimes 3}    \Gamma'(G)^\dagger \Gamma'(G)    \ket{C_+}^{\otimes 3}\\
    &=-    6|\alpha|^2\tanh|\alpha|^2    \left(        \braket{\chi|\chi}        -        |\braket{C_-|\chi}|^2    \right).
\end{aligned}
\end{equation}
The inner product is
\begin{equation}
\begin{aligned}
    \braket{\chi|\chi}
    &=
    \bra{C_+}aa^\dagger\ket{C_+}       \\
    &=
    1+\bra{C_+}a^\dagger a\ket{C_+}    \\
    &=
    1+|\alpha|^2\tanh|\alpha|^2,
\end{aligned}
\end{equation}
and the nontrivial cross term is
\begin{equation}
\begin{aligned}
    \braket{C_-|\chi}
    &=
    \bra{C_-}a^\dagger\ket{C_+}       \\
    &=
    \overline{\alpha}\sqrt{\coth|\alpha|^2}.
\end{aligned}
\end{equation}
Hence, we have 
\begin{equation}
     \begin{aligned}
         \bra{C_+}^{\otimes 3}\Gamma'(G)^\dagger \Gamma'(G) \ket{C_+}^{\otimes 3}&= 6|\alpha|^2\tanh{|\alpha|^2}\Big(1+|\alpha|^2 \tanh{|\alpha|^2}-\frac{|\alpha|^2}{\tanh{|\alpha|^2}}\Big)\\
         &=6 |\alpha|^2(\tanh{|\alpha|^2}-|\alpha|^2 \sech^2 |\alpha|^2 ).
     \end{aligned}
\end{equation}
Therefore,
\begin{equation}
\begin{aligned}
    \frac{{\rm d}^2}{{\rm d}\theta^2}
    \bra{C_+}^{\otimes 3}
    e^{\theta \Gamma'(G)}
    \ket{C_+}^{\otimes 3}
    \bigg|_{\theta=0}
    =
    -6|\alpha|^2
    \left(
        \tanh|\alpha|^2
        -
        |\alpha|^2\sech^2|\alpha|^2
    \right).
\end{aligned}
\end{equation}
For \(|\alpha|^2>0\), the term in parentheses is strictly positive (it vanishes only at $\alpha=0$). Hence the second derivative is nonzero whenever \(\alpha\neq0\). Thus, the even cat state is non-Gaussian, except for the trivial case \(\alpha=0\), where it reduces to the vacuum state.

\textbf{Example 3. A mixture of two coherent states.} Consider a mixed state
\begin{equation}
    \rho_{\mathrm{mix}}=\frac{1}{2}(\ket{\alpha}\bra{\alpha}+\ket{-\alpha}\bra{-\alpha}).
\end{equation}
Note that this is a statistical mixture, in contrast to the even cat state which is a coherent superposition of the same two components. 

 Let $r=|\alpha|^2$, the purity is
\begin{equation}
\begin{aligned}
    \operatorname{Tr}(\rho_{\mathrm{mix}}^2)
    &=
    \frac{1}{4}
    \left(
        2+2|\braket{\alpha|-\alpha}|^2
    \right)                                  \\
    &=
    \frac{1}{2}
    \left(
        1+e^{-4r}
    \right).
\end{aligned}
\end{equation}
The reference has the same purity as the input. Equation~\eqref{eq:pur_ther} therefore gives 
\begin{equation}
    \frac{1-\lambda}{1+\lambda}
    =    \frac{1}{2}    \left(        1+e^{-4r}    \right).
\end{equation}
Solving this equation gives
\begin{equation}
  \lambda    =    \frac{1-e^{-4r}}{3+e^{-4r}}.
\end{equation}

For \(O=e^{\theta G}\) with $G$ as in Eq. \eqref{eq:G}, the Gaussian reference value is given by Eq.~\eqref{eq:k3example},
\begin{equation}\label{eq:W}
    W_3(\lambda(r),\theta)    :=      \frac{(1-\lambda(r))^2}
    {1-2\lambda(r) \cos(\sqrt{3}\theta)+\lambda(r)^2},
\end{equation}
where $\lambda(r)$ denotes the dependence on $r$. 

Next we compute the expectation value for three copies of \(\rho_{\mathrm{mix}}\), define
\begin{equation}
\begin{aligned}
    f_3(\rho_{\mathrm{mix}},\theta)    &:=    \operatorname{Tr}    \left(        \rho_{\mathrm{mix}}^{\otimes 3}\Gamma(e^{\theta G})    \right)                                  \\
    &=    \frac{1}{8}    \sum_{\epsilon_1,\epsilon_2,\epsilon_3=\pm1}
    \bra{\epsilon_1\alpha,\epsilon_2\alpha,\epsilon_3\alpha}
    \Gamma(e^{\theta G}) \ket{\epsilon_1\alpha,\epsilon_2\alpha,\epsilon_3\alpha}.
\end{aligned}
\end{equation}
Let $ \epsilon=(\epsilon_1,\epsilon_2,\epsilon_3)^T.$ Using the overlap formula for coherent states, 
\begin{equation}
    \braket{\beta|\gamma}
    =
    \exp\left(
        -\frac{1}{2}\|\beta\|^2
        -
        \frac{1}{2}\|\gamma\|^2
        +
        \beta^\dagger\gamma
    \right),
\end{equation}
and the fact that $\Gamma(e^{\theta G})\ket{\epsilon\alpha}=   \ket{(e^{\theta G}\epsilon)\alpha}$ since $\Gamma$ preserves vacuum and acts linearly on coherent-state parameters, we find 
\begin{equation}
\begin{aligned}
    \bra{\epsilon\alpha}
    \Gamma(e^{\theta G})
    \ket{\epsilon\alpha} 
    &=
    \exp\left(
        -3r+r\,\epsilon^T e^{\theta G}\epsilon
    \right).
\end{aligned}
\end{equation}
When
\[
    \epsilon=(1,1,1)^T
    \quad\text{or}\quad
    \epsilon=(-1,-1,-1)^T,
\]
one has $\epsilon^T e^{\theta G}\epsilon=3.$ For the remaining six choices of \(\epsilon\), 
\begin{equation}
    \epsilon^T e^{\theta G}\epsilon
    =
    \frac{1}{3}
    +
    \frac{8}{3}
    \cos(\sqrt{3}\theta).
\end{equation}
Therefore,
\begin{equation}\label{eq:f}
    f_3(\rho_{\mathrm{mix}},\theta)
    =    \frac{1}{4}    \left( 1+  3\exp\left(
    -\frac{8r}{3} \left(1-\cos(\sqrt{3}\theta)            \right) \right)  \right).
\end{equation}
To verify that Eqs.~\eqref{eq:W} and \eqref{eq:f} cannot agree as functions of $\theta$ for $r>0$, put $t=1-\cos(\sqrt3\theta)$ and $c=2\lambda/(1-\lambda)^2$. Then $f(t)=1/4+(3/4)e^{-8rt/3}$ and $W(t)=(1+ct)^{-1}$. Equality of their first derivatives at $t=0$ would require $c=2r$, whereas equality of their second derivatives would require $16r^2/3=2c^2=8r^2$, a contradiction for $r>0$. Hence the state is non-Gaussian for $\alpha\neq0$, consistently with Ref.~\cite{PRXQuantum.2.030204}.

For illustration, fix $\theta_0 = 2\pi/(3\sqrt3)$, so that $\cos(\sqrt3\theta_0)=-1/2$. Then
\begin{equation}
 W_3(\lambda(r),\theta_0)
 =\frac{(1-\lambda(r))^2}
 {1+\lambda(r)+\lambda(r)^2},
 \qquad
 \lambda(r)=\frac{1-e^{-4r}}{3+e^{-4r}},
\end{equation}
while
\begin{equation}
 f_3(\rho_{\mathrm{mix}},\theta_0)
 =\frac{1+3e^{-4r}}{4}.
\end{equation}
Consequently, their difference is
\begin{equation}
 \Delta(r,\theta_0)
 =f_3(\rho_{\mathrm{mix}},\theta_0)
  -W_3(\lambda(r),\theta_0)
 =-\frac{3(1-e^{-4r})^3}
 {4(13+2e^{-4r}+e^{-8r})}<0,
 \qquad \forall r>0.
\end{equation}
Thus this single fixed setting certifies ordinary non-Gaussianity for every nonzero amplitude in the
exact-data limit.

\section*{Appendix C: Multimode Generalization}

\subsection*{C.1 Multimode extension of the copy symmetries}
Let $a_{r,\ell}$ annihilate physical mode $r=1,\ldots,d$ in copy $\ell=1,\ldots,k$. For each copy, use the complex operator
ordering of Appendix A,
\[
 \xi_\ell=(a_{1,\ell},\ldots,a_{d,\ell},
 a_{1,\ell}^\dagger,\ldots,a_{d,\ell}^\dagger)^T.
\]
Applying the same orthogonal matrix to the copies of each physical mode defines
\begin{equation}
 \Gamma_d(O)=\bigotimes_{r=1}^d\Gamma^{(r)}(O),
 \qquad O\in O(k).
\end{equation}
The factors act on disjoint modes and commute. Its action is
\begin{equation}\label{eq:multi-copy-action}
 \Gamma_d(O)a_{r,\ell}\Gamma_d(O)^{-1}
 =\sum_{\beta=1}^kO_{\beta\ell}a_{r,\beta}.
\end{equation}
Since $O$ is real, creation operators have the same transformation, and $\Gamma_d(O)\xi_\ell\Gamma_d(O)^{-1}=\sum_\beta O_{\beta\ell}\xi_\beta$.
 Write $f_k^{(d)}(\rho,O)=\operatorname{Tr}(\rho^{\otimes k}\Gamma_d(O))$, with $f_k^{(1)}=f_k$ as in the main text.
\begin{lemma}\label{th:multi_symplectic_commutant}
For any $O\in O(k)$ and any $S\in Sp(2d,\mathbb{R})$, one has
\begin{equation}
 [\Gamma_d(O),\Lambda(S)^{\otimes k}]=0.
\end{equation}
\end{lemma}
\begin{proof}
First take $S=e^Y$ and write $\Lambda(S)=\exp(-i\xi^\dagger h\xi)$, absorbing the factor $1/2$ into $h$. As in Lemma~\ref{appenlem:com_sym}, the doubled vector $\xi_\ell$ is now $2d$-dimensional and the $k$-copy generator is $G'=\sum_{\ell=1}^{k}\xi_\ell^\dagger h\xi_\ell$, so $\Lambda(S)^{\otimes k}=e^{-iG'}$. Using Eq.~\eqref{eq:multi-copy-action} and the orthogonality of $O$,
\begin{equation}
\begin{aligned}
    \Gamma_d(O)G'\Gamma_d(O)^{-1}
    &=    \sum_{\ell,\beta,\gamma=1}^{k}
    O_{\beta\ell}O_{\gamma\ell}\,\xi_\beta^\dagger h\xi_\gamma
    =    \sum_{\beta,\gamma=1}^{k}
    \delta_{\beta\gamma}\,\xi_\beta^\dagger h\xi_\gamma
    =    G'.
\end{aligned}
\end{equation}
Exponentiating gives $\Gamma_d(O)e^{-iG'}\Gamma_d(O)^{-1}=e^{-iG'}$. A general $S\in\operatorname{Sp}(2d,\mathbb R)$ is a finite product of symplectic exponentials. Applying the same identity to each factor proves $[\Gamma_d(O),\Lambda(S)^{\otimes k}]=0$ for every $S$, independently of the overall phase of its implementation.
\end{proof}

A common multimode displacement is
\begin{equation}
 D(\boldsymbol\alpha)^{\otimes k}
 =\exp\!\left(
 \sum_{\ell=1}^k\sum_{r=1}^d
 (\alpha_r a_{r,\ell}^\dagger-\overline\alpha_r a_{r,\ell})
 \right).
\end{equation}
For every physical mode it drives only the collective copy combination $k^{-1/2}\sum_\ell a_{r,\ell}$.
\begin{lemma}\label{lem:multi_displacement_commutant}
Suppose $O\in O(k)$ and $\boldsymbol{\alpha}\neq 0$. Then
\begin{equation}
[\Gamma_d(O),D(\boldsymbol{\alpha})^{\otimes k}]=0
\end{equation}
if and only if
\begin{equation}
O\mathbf e_1=\mathbf e_1.
\end{equation}
\end{lemma}
\begin{proof}
This is the multimode analogue of Lemma~\ref{appenlem:dis_com}. With $c_\beta=\sum_{\ell=1}^{k}O_{\beta\ell}$, Eq.~\eqref{eq:multi-copy-action} gives
\begin{equation}
    \Gamma_d(O)D(\boldsymbol{\alpha})^{\otimes k}\Gamma_d(O)^{-1}
    =    \bigotimes_{\beta=1}^{k}D(c_\beta\boldsymbol{\alpha}).
\end{equation}
By the faithfulness of the Weyl representation, $D(\boldsymbol u)=D(\boldsymbol v)$ implies $\boldsymbol u=\boldsymbol v$. Hence $[\Gamma_d(O),D(\boldsymbol{\alpha})^{\otimes k}]=0$ if and only if $c_\beta\boldsymbol{\alpha}=\boldsymbol{\alpha}$ for every $\beta$, which, as $\boldsymbol{\alpha}\neq\boldsymbol 0$, is equivalent to $c_\beta=1$ for every $\beta$, i.e.\ $O\boldsymbol e_1=\boldsymbol e_1$.
\end{proof}

\begin{lemma}\label{appth:multi_jacobi_commutant}
For $\boldsymbol\alpha\ne0$ and $S\in\operatorname{Sp}(2d,\mathbb R)$, an orthogonal copy mixer $\Gamma_d(O)$ commutes with $(D(\boldsymbol\alpha)\Lambda(S))^{\otimes k}$ if and only if
\begin{equation}\label{eq:multi_thcon}
 O\boldsymbol e_1=\boldsymbol e_1,
 \qquad O\in O(k).
\end{equation}
The admissible orthogonal copy mixers form the subgroup $\widetilde G_k\cong O(k-1)$.
\end{lemma}
\begin{proof}
Lemma~\ref{th:multi_symplectic_commutant} cancels the common symplectic factor from the commutation condition, and Lemma~\ref{lem:multi_displacement_commutant} supplies the
displacement condition. The block decomposition $C^TOC=1\oplus\widetilde O$ from B.3 gives the subgroup.
\end{proof}

\subsection*{C.2 Proof of Theorem 3 }
By Williamson's decomposition, every $d$-mode Gaussian state has the form
\begin{equation}
 \rho_g=D(\boldsymbol\alpha)\Lambda(S)
 \rho_{\boldsymbol\lambda}
 \Lambda(S)^\dagger D(\boldsymbol\alpha)^\dagger,
 \qquad
 \rho_{\boldsymbol\lambda}
 =\bigotimes_{r=1}^d\rho_{\lambda_r},
\end{equation}
where $0\leq\lambda_r<1$ are the Williamson thermal parameters. For $O\in\widetilde G_k$, putting $V_g=D(\boldsymbol\alpha)\Lambda(S)$ gives
\begin{equation}
\begin{aligned}
 \operatorname{Tr}\big(\rho_g^{\otimes k}\Gamma_d(O)\big)
 &=\operatorname{Tr}[\rho_{\boldsymbol\lambda}^{\otimes k}
 V_g^{\dagger\otimes k}\Gamma_d(O)V_g^{\otimes k}]\\
 &=\operatorname{Tr}[\rho_{\boldsymbol\lambda}^{\otimes k}\Gamma_d(O)].
\end{aligned}
\end{equation}
Thus the Gaussian preparation cancels even when it couples physical modes and only the thermal parameters remain.

\begin{theorem}\label{appenth:multiref}
For any \(O\in \widetilde{G}_k\), the multimode Gaussian reference value is 
\begin{equation}\label{eq:multi_gaussian_reference} 
W_k(\boldsymbol{\lambda},O) := \operatorname{Tr}\!\left(   \rho_{\boldsymbol{\lambda}}^{\otimes k}\Gamma_d(O) \right) = \prod_{r=1}^{d} \frac{(1-\lambda_r)^k} {\det(I_k-\lambda_r O)}.
\end{equation}
\end{theorem}
\begin{proof}
Regrouping the tensor factors according to physical mode gives
\[
 \rho_{\boldsymbol\lambda}^{\otimes k}
 =\bigotimes_{r=1}^d\rho_{\lambda_r}^{\otimes k},
 \qquad
 \Gamma_d(O)=\bigotimes_{r=1}^d\Gamma^{(r)}(O).
\]
The trace therefore factors into $d$ single-mode references,
\begin{equation}\label{eq:multi_gaussian_reference-factorized}
\begin{aligned}
 W_k(\boldsymbol\lambda,O)
 &=\prod_{r=1}^d
 \operatorname{Tr}\big(\rho_{\lambda_r}^{\otimes k}\Gamma^{(r)}(O)\big)\\
 &=\prod_{r=1}^d
 \frac{(1-\lambda_r)^k}{\det(I_k-\lambda_r O)},
\end{aligned}
\end{equation}
where the last equality is Theorem~\ref{appenth:refval}.
\end{proof}

For the interleaved quadrature ordering $\boldsymbol R$ and $\Omega$, the symplectic eigenvalues of the CM $\sigma_\rho$ are the positive eigenvalues of $i\Omega\sigma_\rho$ \cite{adesso2014continuous}. Specifically,
\[
\operatorname{spec}(i\Omega\sigma_\rho)
=\{+\nu_1,-\nu_1,\ldots,+\nu_d,-\nu_d\}.
\]
Thus, one can obtain $\nu_1,\ldots,\nu_d$ by an ordinary eigenvalue calculation, without explicitly constructing the
Williamson transformation. An equivalent numerical procedure is to diagonalize the Hermitian matrix $H=i\sigma_\rho^{1/2}\Omega\sigma_\rho^{1/2}$ and retain its positive eigenvalues.

Once the symplectic eigenvalues are known, the thermal parameters of the covariance-matched Gaussian reference are
\[
\lambda_r=\frac{\nu_r-1}{\nu_r+1},
\qquad r=1,\ldots,d.
\]
These calculations apply to the CM of an arbitrary state, including a non-Gaussian state. With the convention that the vacuum covariance is $I_{2d}$, every physical CM satisfies $\nu_r\geq1$. The reference can be evaluated directly from these symplectic eigenvalues. Indeed,
\begin{equation}
 I_k-\lambda_rO
 =\frac{(\nu_r+1)I_k-(\nu_r-1)O}{\nu_r+1},
 \qquad 1-\lambda_r=\frac{2}{\nu_r+1},
\end{equation}
and cancellation of the factors $(\nu_r+1)^{-k}$ gives
\begin{equation}\label{eq:reference-from-covariance-sm}
 W_k(\boldsymbol\lambda,O)
 =\prod_{r=1}^{d}
 \frac{2^k}{\det[(\nu_r+1)I_k-(\nu_r-1)O]}.
\end{equation}

If $\rho$ were Gaussian, the symmetry would give the reference for every admissible mixer. Consequently,
\begin{equation}
 \operatorname{Tr}\big(\rho^{\otimes k}\Gamma_d(O)\big)
 \ne W_k(\boldsymbol\lambda,O),
 \qquad O\in\widetilde G_k,
\end{equation}
certifies non-Gaussianity. 

\subsection*{C.3 Two-mode Fock-state example}

Consider $\rho=|1,1\rangle\langle1,1|$ and $k=3$ copies. The first moments vanish and the covariance is $3I_4$, so $\nu_1=\nu_2=3$ and $\lambda_1=\lambda_2=1/2$. This is a covariance-matched reference and its purity need not equal the purity of the non-Gaussian input.

Take $O=e^{\theta G}$ with $G$ from Eq.~\eqref{eq:G} and $\theta=\pi/(2\sqrt3)$. Then
\begin{equation}\label{eq:O3}
 O=\begin{pmatrix}x&y&z\\z&x&y\\y&z&x\end{pmatrix},
 \qquad x=\frac13,\quad
 y=\frac{1-\sqrt3}{3},\quad z=\frac{1+\sqrt3}{3}.
\end{equation}
Since the eigenvalues are $1,i,-i$,
\[
 \det(I_3-\lambda O)=(1-\lambda)(1+\lambda^2),
\]
and the reference is
\[
 W_3(\boldsymbol\lambda,O)
 =\left(\frac{(1-\lambda)^2}{1+\lambda^2}\right)^2_{\lambda=1/2}
 =\frac1{25}.
\]
Regrouping the six input modes by their physical-mode index factorizes the actual expectation,
\[
 \operatorname{Tr}\big(\rho^{\otimes3}\Gamma_2(O)\big)
 =\left(\langle 1|^{\otimes 3}\Gamma(O)|1\rangle^{\otimes3}\right)^2.
\]
The single-mode factor is the permanent of $O$,
\[
\langle 1|^{\otimes 3}\Gamma(O)|1\rangle^{\otimes3}
 =\operatorname{Per}(O)=x^3+y^3+z^3+3xyz=\frac59.
\]
Hence
\[
 \operatorname{Tr}\big(\rho^{\otimes3}\Gamma_2(O)\big)
 =\frac{25}{81}\ne\frac1{25}
 =W_3(\boldsymbol\lambda,O),
\]
which certifies non-Gaussianity.

\section*{Appendix D: Completeness of Orthogonal-Copy Gaussianity Tests}
\label{sec:completeness}
We prove main-text Theorem 4 for the covariance-calibrated reference family of multimode states, and then prove completeness of the purity-calibrated single-mode family. These are different calibrations for a general non-Gaussian input. At fixed $d$, the notation $\mathcal S_k$ in the main text denotes the sets $\mathcal S_k^{(d)}$ defined below.

\begin{definition}
Let $\nu_1,\ldots,\nu_d$ be the symplectic eigenvalues of the CM of a $d$-mode density operator $\rho$, and set
$$
\lambda_r=\frac{\nu_r-1}{\nu_r+1},\qquad
\boldsymbol\lambda=(\lambda_1,\ldots,\lambda_d).
$$
For $k\geq1$, define the set of states passing the $k$-copy reference tests as
\begin{equation}\label{eq:Sk-exact}
\mathcal S_k^{(d)}
=\left\{\rho:\sum_{r=1}^{d}\operatorname{Tr}\!\left(\rho \hat{n}_r\right)<\infty,
\quad \operatorname{Tr}\!\left(\rho^{\otimes k}\Gamma_d(O)\right)
=W_k(\boldsymbol\lambda,O),\forall O\in\widetilde G_k\right\}.
\end{equation}
\end{definition}
\begin{theorem}[Completeness]\label{appth:complete-Gaussian}
The sets satisfy
\begin{equation}\label{eq:complete-intersection}
\mathcal S_{k+1}^{(d)}\subseteq\mathcal S_k^{(d)},
\qquad
\bigcap_{k=1}^{\infty}\mathcal S_k^{(d)}=\mathcal G_d,
\end{equation}
where $\mathcal G_d$ is the set of $d$-mode Gaussian states. 
\end{theorem}

\begin{proof}
We first prove nestedness. If $O\in\widetilde G_k$, then $O\oplus1\in\widetilde G_{k+1}$, and the extra copy is left unchanged. Since $\operatorname{Tr}\rho=1$,
\begin{equation}
\operatorname{Tr}\!\left(\rho^{\otimes(k+1)}\Gamma_d(O\oplus1)\right)
=\operatorname{Tr}\!\left(\rho^{\otimes k}\Gamma_d(O)\right).
\end{equation}
For every $r$, the determinant factorizes as
\begin{equation}
\det\!\left(I_{k+1}-\lambda_r(O\oplus1)\right)
=(1-\lambda_r)\det(I_k-\lambda_r O),
\end{equation}
so the reference values agree, $W_{k+1}(\boldsymbol\lambda,O\oplus1)=W_k(\boldsymbol\lambda,O)$. Therefore $\mathcal S_{k+1}^{(d)}\subseteq\mathcal S_k^{(d)}$. Conversely, every Gaussian state satisfies all reference identities by the multimode thermal-reference formula, hence $\mathcal G_d\subseteq\bigcap_k\mathcal S_k^{(d)}$.

It remains to show that any $\rho\in\bigcap_k\mathcal S_k^{(d)}$ is Gaussian. Let $\tau_\rho$ be the unique Gaussian state with the same first moments and covariance as $\rho$. By Williamson's decomposition it is Gaussian-unitary equivalent to $\rho_{\boldsymbol\lambda}=\bigotimes_{r=1}^d\rho_{\lambda_r}$. Consider the $k$-cycle permutation $P_k$, the $k\times k$ matrix that cyclically shifts $(v_1,\ldots,v_k)^T\mapsto(v_k,v_1,\ldots,v_{k-1})^T$. Since $P_k$ fixes the all-ones vector, $P_k\in\widetilde G_k$. Moreover, $\Gamma_d(P_k)$ cyclically permutes the $k$ copies of the entire $d$-mode system, correlations between physical modes included. Because the cyclic shift merely relabels the $k$ tensor factors, the trace reduces to
\begin{equation}
\operatorname{Tr}\!\left(\rho^{\otimes k}\Gamma_d(P_k)\right)=\operatorname{Tr}(\rho^k).
\end{equation}
Since $\det(I_k-\lambda_rP_k)=1-\lambda_r^k$, the reference value at $O=P_k$ is
\begin{equation}
W_k(\boldsymbol\lambda,P_k)
=\prod_{r=1}^d\frac{(1-\lambda_r)^k}{1-\lambda_r^k}
=\operatorname{Tr}(\tau_\rho^k).
\end{equation}
As $\rho\in\mathcal S_k^{(d)}$, the two sides coincide, so
\begin{equation}\label{eq:cycle-spectral-moments}
\operatorname{Tr}(\rho^k)=\operatorname{Tr}(\tau_\rho^k)
\qquad\forall k\geq 1.
\end{equation} 

To pass from moments to the spectrum, let $r_j>0$ be the nonzero eigenvalues of an arbitrary state $\sigma$, counted with multiplicity, and let $\mu_\sigma=\sum_jr_j\delta_{r_j}$ be the associated measure on $[0,1]$. Its moments are
\begin{equation}
\int_0^1x^\ell\,{\rm d}\mu_\sigma(x)=\sum_j r_j \cdot r_j^l=\operatorname{Tr}(\sigma^{\ell+1}),
\qquad\ell=0,1,\ldots
\end{equation}
By Eq.~\eqref{eq:cycle-spectral-moments}, the measures $\mu_\rho$ and $\mu_{\tau_\rho}$ coincide since polynomials are dense in $C([0,1])$. For $x>0$, the multiplicity of $x$ as an eigenvalue equals $\mu_\rho(\{x\})/x$, so $\rho$ and $\tau_\rho$ have the same nonzero spectrum, including multiplicities. Consequently,
\begin{equation}
S(\rho)=S(\tau_\rho)<\infty,
\end{equation}
where $S$ is the von Neumann entropy and the finiteness follows from that $\tau_\rho$ is a finite product of thermal states with $0\leq\lambda_r<1$.
First, suppose that all Williamson parameters satisfy $\nu_r>1$, so that $\tau_\rho$ is faithful, i.e. $\tau_\rho>0$. Choose a Gaussian unitary $V_\rho$ with $\tau_\rho=V_\rho\rho_{\boldsymbol\lambda}V_\rho^\dagger$. Then
\begin{equation}
\begin{aligned}
    \log\tau_\rho&=\Big(\sum_{r=1}^d\log(1-\lambda_r)\Big)I
+\sum_{r=1}^d(\log\lambda_r)\,V_\rho a_r^\dagger a_rV_\rho^\dagger\\
&=cI+\sum_{j=1}^{2d}\ell_jR_j+\sum_{j,k=1}^{2d}Q_{jk}R_jR_k,
\end{aligned}
\end{equation}
which is a polynomial of degree at most two in the quadratures and $R_1,\ldots,R_{2d}$ are the components of the quadrature vector $\boldsymbol{R}$. Its expectation is finite by the second-moment assumption, and equal first and second moments give
\begin{equation}
\operatorname{Tr}(\rho\log\tau_\rho)=\operatorname{Tr}(\tau_\rho\log\tau_\rho).
\end{equation}
Using $S(\rho)=S(\tau_\rho)$,
\begin{equation}\label{eq:entropy-completeness}
D(\rho\Vert\tau_\rho)
=S(\tau_\rho)-S(\rho)=0.
\end{equation}
By the equality condition for the quantum relative entropy, $\rho=\tau_\rho$. This is the strict Gaussian maximum-entropy property at fixed first and second moments \cite{Genoni2008}.

It remains to remove the assumption $\nu_r>1$. If some $\nu_r=1$, apply $V_\rho^\dagger$ and set $\rho'=V_\rho^\dagger\rho V_\rho$. Since $\rho$ and $\tau_\rho$ share first moments and covariance, and $V_\rho^\dagger\tau_\rho V_\rho=\rho_{\boldsymbol\lambda}$, the state $\rho'$ has vanishing first moments and covariance $\bigoplus_r\nu_rI_2$, i.e.\ Williamson form. For a mode with $\nu_r=1$, the covariance $\nu_rI_2$ gives $\langle q_r^2\rangle=\langle p_r^2\rangle=\nu_r/2$. Via $q_r^2+p_r^2=2a_r^\dagger a_r+1$, this yields
\begin{equation}
\operatorname{Tr}(\rho'a_r^\dagger a_r)
=\frac{\langle q_r^2\rangle+\langle p_r^2\rangle-1}{2}
=\frac{\nu_r-1}{2}=0 .
\end{equation}
Because $\hat n_r\geq0$ and $\rho'\geq0$, the vanishing of $\operatorname{Tr}(\rho'\hat n_r)=\|\hat n_r^{1/2}\sqrt{\rho'}\|_2^2$ forces $\hat n_r\rho'=0$, so $\rho'$ is supported on the kernel $\ker\hat n_r=\mathrm{span}\{|0\rangle_r\}$. The reduced state on mode $r$ is therefore pure, and $\rho'$ factorizes as $|0\rangle\langle0|_r\otimes\sigma_{\mathrm{rest}}$. Removing all such vacuum factors from $\rho'$ and from $\rho_{\boldsymbol\lambda}$ leaves states with equal first and second moments and equal nonzero spectra, while the remaining thermal reference is faithful. The preceding relative-entropy argument applies to these modes. If nothing remains, $\rho'$ is the full vacuum. Hence $\rho=\tau_\rho$ in this case as well, which proves Eq.~\eqref{eq:complete-intersection}.

\end{proof}

\section*{Appendix E: Quantum Non-Gaussianity from Real-Part Data}
\label{sec:genuine}

This appendix distinguishes the Gaussian set from its convex hull and proves the real-part bounds and joint-moment hierarchy used in the main text. The physical modes and the copy number are denoted by $d$ and $k$, respectively. The moment level is denoted by $L$ and increasing $L$ does not increase $k$. Each copy contains the entire $d$-mode system, and we use the copy mixer $\Gamma_d(O)=\prod_{r=1}^d\Gamma^{(r)}(O)$.

\subsection*{E.1 Gaussian mixtures and the invariant sector}
Let $\mathcal G_d$ be the set of $d$-mode Gaussian states and
\[
 \mathcal C_{\rm g}
 =\overline{\operatorname{conv}}(\mathcal G_d)
 =\overline{\operatorname{conv}}
 \{|g\rangle\langle g|:|g\rangle\ \text{pure Gaussian}\},
\]
where both closures are in trace norm. The two convex hulls agree since a product of thermal states is a mixture of multimode coherent states, and a Gaussian unitary sends every such state to a pure Gaussian state. Ordinary non-Gaussianity excludes $\mathcal G_d$, but quantum non-Gaussianity excludes $\mathcal C_{\rm g}$ \cite{Filip2011,Geno2013}.

Fix $O\in\widetilde G_k=\{O\in O(k):O\boldsymbol e_1=\boldsymbol e_1\}$, and put $p_\rho=\operatorname{Tr}\rho^2$. Here $p_\rho$ is the purity of the entire $d$-mode state. Let $\Pi_O$ be the spectral projector of $\Gamma_d(O)$ at eigenvalue one. If $|g\rangle=V|0\rangle$ for a $d$-mode Gaussian unitary $V$, with $|0\rangle$ denoting the $d$-mode vacuum, then
\begin{equation}\label{eq:pure-G-invariance-sm}
 \Gamma_d(O)|g\rangle^{\otimes k}
 =V^{\otimes k}\Gamma_d(O)|0\rangle^{\otimes k}
 =|g\rangle^{\otimes k}.
\end{equation}
Thus every identical pure Gaussian tensor lies in the invariant sector, although a product of independently chosen, different Gaussian components generally does not. 

\begin{theorem}\label{th:genuine-real-sm}
For any $\rho\in\mathcal C_{\rm g}$, $k\geq2$, and $O\in\widetilde{G}_k$,
\begin{equation}\label{eq:genuine-main}
 \gamma_1(\rho)=\operatorname{Re}\Big( {\rm Tr}\big(\rho^{\otimes k}\Gamma_d(O)\big) \Big) 
 \geq 2p_\rho^{\,k-1/2}-1,
\end{equation}
 where $p_\rho=\operatorname{Tr}(\rho^2)$ is the actual global purity of the complete $d$-mode input. The criterion uses this purity and the measured real expectation $\gamma_1$, without covariance estimation or a finite-energy assumption. A violation certifies quantum non-Gaussianity.
\end{theorem}

We prove this theorem through the following stronger supplementary bounds. The operator norm $\ell_\rho=\|\rho\|$ is the largest eigenvalue of $\rho$. These statements require no photon-number cutoff or finite energy assumption.
\begin{proposition}
\label{prop:sector-bound-sm}
For every $\rho\in\mathcal C_{\rm g}$, $k\geq2$, and $O\in\widetilde G_k$,
\begin{equation}\label{eq:sharp-sector-sm}
 q_O:=\operatorname{Tr}(\rho^{\otimes k}\Pi_O)
 \geq\frac{p_\rho^k}{\ell_\rho}\geq p_\rho^{\,k-1/2}\geq p_\rho^k,
\end{equation}
and
\begin{equation}\label{eq:sharp-real-sm}
 \gamma_1
 \geq 2q_O-1\geq\frac{2p_\rho^k}{\ell_\rho}-1
 \geq2p_\rho^{\,k-1/2}-1,
\end{equation}
where  $\ell_\rho=\|\rho\|$ is the largest eigenvalue of $\rho$.
\end{proposition}

\begin{proof}
First take a finite pure Gaussian mixture $\rho=\sum_j w_j P_j$, where $P_j=|g_j\rangle\langle g_j|$, $w_j\geq0$, and $\sum_j w_j=1$. Each $P_j$ describes a pure Gaussian state of the full $d$-mode system and may be entangled across its physical modes. Define
\[
 \tau_k=\sum_j w_j P_j^{\otimes k},\qquad
 F_j=\operatorname{Tr}(\rho P_j).
\]
By Eq.~\eqref{eq:pure-G-invariance-sm}, $\tau_k$ is supported in $\Pi_O$. Product factorization across the copies and Jensen's inequality yield
\begin{equation}\label{eq:Jensen-purity-sm}
 \operatorname{Tr}(\rho^{\otimes k}\tau_k)
 =\sum_j w_j F_j^k
 \geq\Big(\sum_jw_j F_j\Big)^k=p_\rho^k.
\end{equation}
Since $0\leq P_j\leq I$,
\[
 \tau_k\leq\rho\otimes I^{\otimes(k-1)}
 \leq\ell_\rho I.
\]
Combining this with its support gives $\tau_k\leq\ell_\rho\Pi_O$. Therefore, $p^k\leq\ell_\rho q_O$, which proves the first bound. Since $\ell_\rho^2\leq p_\rho$ (the largest eigenvalue squared is bounded by the sum of squares) and $p_\rho\leq1$, we also have $p_\rho^k/\ell_\rho\geq p_\rho^{\,k-1/2}\geq p_\rho^k$, completing Eq.~\eqref{eq:sharp-sector-sm}.

We next prove Eq. \ref{eq:sharp-real-sm} for finite pure Gaussian mixtures. Since $\Gamma_d(O)$ is unitary and $\Pi_O$ is its $+1$ eigenspace projection, it admits the block decomposition $\Gamma_d(O)=\Pi_O+(I-\Pi_O)\Gamma_d(O)(I-\Pi_O)$. Hence, with $B=(I-\Pi_O)\rho^{\otimes k}(I-\Pi_O)\geq0$ and $\operatorname{Tr}B=1-q_O$,
\begin{equation}
\gamma_1=\operatorname{Re}\operatorname{Tr}(\rho^{\otimes k}\Gamma_d(O))
=q_O+\operatorname{Re}\operatorname{Tr}(B\Gamma_d(O))
\geq q_O-\operatorname{Tr}B
=2q_O-1,
\end{equation}
using $|\operatorname{Tr}(B\Gamma_d(O))|\leq\operatorname{Tr}B$, which holds since $\Gamma_d(O)$ is unitary. Combined with Eq.~\eqref{eq:sharp-sector-sm}, this gives Eq.~\eqref{eq:sharp-real-sm}.

Finally, we extend the result to the closed Gaussian convex hull. Every mixed Gaussian state is a mixture of pure Gaussian states, so finite pure Gaussian mixtures are trace-norm dense in $\mathcal C_{\rm g}$. For any $\rho\in\mathcal C_{\rm g}$, choose such mixtures $\rho_n$ with $\|\rho_n-\rho\|_1\to0$. The estimates
\[
 \|\rho_n^{\otimes k}-\rho^{\otimes k}\|_1
 \leq k\|\rho_n-\rho\|_1,
 \qquad
 \bigl|\operatorname{Tr}(\rho_n^2)-p_\rho\bigr|
 \leq2\|\rho_n-\rho\|_1,
\]
and
\[
 |\ell_{\rho_n}-\ell_\rho|
 \leq\|\rho_n-\rho\|_1
\]
show that the tensor powers, purities, and operator norms converge. Since $\Pi_O$ and $\Gamma_d(O)$ are bounded, their expectation values also converge. Moreover, $\ell_\rho>0$ for every density operator. Thus the finite-mixture inequalities pass to the limit, proving both bounds for every $\rho\in\mathcal C_{\rm g}$.
\end{proof}
Physically, Eq.~\eqref{eq:sharp-sector-sm} asserts that every Gaussian preparation carries a minimal invariant-sector weight, while Eq.~\eqref{eq:sharp-real-sm} asserts that the residual weight outside the sector can suppress the real interference signal by at most its own magnitude.

\subsection*{E.2 Nested hierarchy using only real moments}
For a fixed $k\geq2$ and $O\in\widetilde G_k$, define
\[
 \gamma_s=\operatorname{Re}\!\left(
 \operatorname{Tr}\!\left(\rho^{\otimes k}\Gamma_d(O)^s\right)\right),\qquad
 \gamma_0=1,\qquad
 C_L=(\gamma_{|j-i|})_{i,j=0}^L,\qquad
 J_L=\boldsymbol1\boldsymbol1^T.
\]
As in the main text, the hierarchy below fixes $b=p^{k-1/2}_\rho$. 

\begin{theorem}[Joint real-moment hierarchy]
\label{th:real-hierarchy-sm}
Every $\rho\in\mathcal C_{\rm g}$ satisfies, for $b=p^{k-1/2}_\rho$,
\begin{equation}\label{eq:real-hierarchy-sm}
 R_L(\rho):=C_L-bJ_L\succeq0,\qquad L=0,1,\cdots.
\end{equation}
The constraints are nested in $L$, and a violation at any level certifies quantum non-Gaussianity.
\end{theorem}
As in the $s=1$ case, Gaussianity enters only through the invariant sector, where the eigenvalue $z=1$ of $\Gamma_d(O)$ carries some weight not less than $ b$, and this single atom constrains all moments $\gamma_s$ simultaneously.
\begin{proof}
For real coefficients $c=(c_0,\dots,c_L)^T$, define the polynomial $P_c(z)=\sum_{j=0}^L c_jz^j$. By unitarity, $\Gamma_d(O)^\dagger=\Gamma_d(O)^{-1}$, so
\begin{equation}
P_c(\Gamma_d(O))^\dagger P_c(\Gamma_d(O))
=\sum_{i,j=0}^{L}c_ic_j\,\Gamma_d(O)^{j-i}.
\end{equation}
Tracing against $\rho^{\otimes k}$, the diagonal terms $i=j$ contribute $\sum_i c_i^2$, while for each pair $i<j$ the two cross terms are complex conjugates,
\begin{equation}
\operatorname{Tr}\!\big(\rho^{\otimes k}\Gamma_d(O)^{i-j}\big)
=\operatorname{Tr}\!\big(\rho^{\otimes k}(\Gamma_d(O)^{j-i})^\dagger\big)
=\overline{\operatorname{Tr}\!\big(\rho^{\otimes k}\Gamma_d(O)^{j-i}\big)},
\end{equation}
so their sum is $2\operatorname{Re}\operatorname{Tr}\!\big(\rho^{\otimes k}\Gamma_d(O)^{j-i}\big)=2\gamma_{|i-j|}$. Hence
\begin{equation}
c^TC_Lc
=\operatorname{Tr}\!\big(\rho^{\otimes k}P_c(\Gamma_d(O))^\dagger P_c(\Gamma_d(O))\big)
\geq0,
\end{equation}
where the inequality holds for every state because $\rho^{\otimes k}\geq0$ and $P_c(\Gamma_d(O))^\dagger P_c(\Gamma_d(O))\geq0$.

Let $\nu_{\rho,O}$ be the spectral probability measure of the unitary $\Gamma_d(O)$ in the state $\rho^{\otimes k}$, characterized by
\begin{equation}
\int f(z)\,d\nu_{\rho,O}(z)
=\operatorname{Tr}\!\big(\rho^{\otimes k} f(\Gamma_d(O))\big)
\end{equation}
for every bounded measurable function $f$. Then the equality above then reads
\begin{equation}
c^TC_Lc=\int|P_c(z)|^2\,d\nu_{\rho,O}(z).
\end{equation}
The eigenvalue $z=1$ carries the atom $q_O=\operatorname{Tr}(\rho^{\otimes k}\Pi_O)\geq b$ by Eq.~\eqref{eq:sharp-sector-sm}. Splitting it off and then,
\begin{equation}
c^TC_Lc
=q_O|P_c(1)|^2+\int_{z\neq1}|P_c(z)|^2\,d\nu_{\rho,O}(z),
\end{equation}
where the second term is nonnegative. Since $c^TJ_Lc=\bigl(\sum_{j=0}^Lc_j\bigr)^2=|P_c(1)|^2$,
\begin{equation}
c^TR_L(b)c
=c^TC_Lc-b\,c^TJ_Lc
=(q_O-b)|P_c(1)|^2+\int_{z\neq1}|P_c(z)|^2\,d\nu_{\rho,O}(z)
\geq0.
\end{equation}
Thus $c^TR_L(b)c\geq0$ for every real $c$. Since $R_L(b)$ is real symmetric, this is equivalent to $R_L(b)\succeq0$, proving Eq.~\eqref{eq:real-hierarchy-sm}. 

Finally, $C_L$ and $J_L$ are the upper left principal blocks of $C_{L+1}$ and $J_{L+1}$, so $R_L(b)$ is a principal submatrix of $R_{L+1}(b)$. Hence, $R_{L+1}(b)\succeq0$ implies $R_L(b)\succeq0$, which is the nesting in $L$.
\end{proof}

Just like the $s=1$ case, the supplementary sector bound also permits $R_L(b')\succeq0$ with $b'=p_\rho^k/\ell_\rho$ when the operator norm is known. This choice is distinct from the main-text hierarchy, which uses the purity-based threshold $b=p_\rho^{k-1/2}$.

We illustrate the hierarchy at the first two nontrivial levels. At $L=1$, writing $x=\gamma_1$,
\begin{equation}
 R_1=\begin{pmatrix}
 1-b&x-b\\
 x-b&1-b
 \end{pmatrix}\succeq0,
\end{equation}
which holds if and only if $1-b\geq0$ and
$(1-b)^2\geq(x-b)^2$, i.e.
\begin{equation}\label{eq:L1-bound-sm}
 2b-1\leq\gamma_1\leq1.
\end{equation}
Replacing $\Gamma_d(O)$ by $\Gamma_d(O)^s$ in the same
argument gives the individual constraints
\begin{equation}\label{eq:individual-bounds-sm}
 \gamma_s\geq2b-1,\qquad s=1,2,\ldots.
\end{equation}

The first genuinely joint constraint appears at $L=2$. With $x=\gamma_1$ and $y=\gamma_2$,
\begin{equation}
 R_2=\begin{pmatrix}
 1-b&x-b&y-b\\
 x-b&1-b&x-b\\
 y-b&x-b&1-b
 \end{pmatrix}.
\end{equation}
Since $R_2$ is invariant under the exchange of the first
and third indices, it block diagonalizes in the orthonormal basis
\begin{equation}
 e_+=\tfrac{1}{\sqrt2}(1,0,1)^T,\quad
 e_0=(0,1,0)^T,\quad
 e_-=\tfrac{1}{\sqrt2}(1,0,-1)^T.
\end{equation}
Hence, $R_2\succeq0$ if and only if
$(1-b)(1+y-2b)\geq2(x-b)^2$, together with $1-y\geq0$, $1-b\geq0$, and $1+y-2b\geq0$.

The individual bound on $\gamma_2$ follows from positivity of the principal submatrix of $R_2$ with moment indices $0$ and $2$ (that is, the first and third rows and columns), whereas positivity of the full matrix imposes an additional joint constraint coupling $\gamma_1$ and $\gamma_2$.

We now apply these two levels to the $d$-mode state
\begin{equation}\label{eq:rho-eta-sm}
 \rho_\eta
 =\bigl((1-\eta)\ket{0}\bra{0}
       +\eta\ket{1}\bra{1}\bigr)
 \otimes\bigl(\ket{0}\bra{0}\bigr)^{\otimes(d-1)},
 \qquad 0\leq\eta\leq1.
\end{equation}
For the copy mixer $O_\phi=e^{(\phi/\sqrt3)G}$,
with $G$ from Eq.~\eqref{eq:G}, take $\phi=\pi/2$.
Then $\Gamma_d(O_\phi)^4=I$, so the real moments
\begin{equation}
 \gamma_s=\operatorname{Re}\operatorname{Tr}\bigl(
 \rho_\eta^{\otimes3}\Gamma_d(O_\phi)^s\bigr),
 \qquad s=0,1,\ldots,
\end{equation}
repeat with period four. The sharpened threshold at
$3$ copies is $b=p_\rho^{k-1/2}=p_\rho^{5/2}$,
with purity $p_\rho=1-2\eta+2\eta^2$.

A direct calculation gives
\begin{equation}\label{eq:moments-eta-sm}
 \gamma_1(\eta)
 =1-2\eta+\tfrac23\eta^2+\tfrac89\eta^3,
 \qquad
 \gamma_2(\eta)
 =1-4\eta+\tfrac{20}{3}\eta^2-\tfrac{32}{9}\eta^3.
\end{equation}
At level $L=1$, the detection boundary is determined by saturation of the bound $\gamma_1\geq2b-1$, which gives
\begin{equation}\label{eq:L1-threshold-sm}
 1-2\eta+\tfrac23\eta^2+\tfrac89\eta^3
 =2\bigl(1-2\eta+2\eta^2\bigr)^{5/2}-1.
\end{equation}
The unique nontrivial root in the interval $(0,1)$ is
$\eta_c^{(1)}\simeq0.93372$. Violation occurs for $\eta>\eta_c^{(1)}$.

Applying the individual bound to $\gamma_2$
instead gives the saturation equation
\begin{equation}\label{eq:eta-y-threshold-sm}
 1-4\eta+\tfrac{20}{3}\eta^2-\tfrac{32}{9}\eta^3
 =2\bigl(1-2\eta+2\eta^2\bigr)^{5/2}-1.
\end{equation}
The equation has exactly one nontrivial root in the interval $(0,1)$, namely $\eta_y\simeq0.898509126$.
An violation occurs for $\eta>\eta_y$. Thus the individual $\gamma_2$ test improves on the individual $\gamma_1$ test. Both individual tests retain the copy number $k=3$.

At level $L=2$, the joint detection boundary
is determined by the saturation equation
\begin{equation}
 \bigl(1-b(\eta)\bigr)
 \bigl(1+\gamma_2(\eta)-2b(\eta)\bigr)
 =2\bigl(\gamma_1(\eta)-b(\eta)\bigr)^2.
\end{equation}
Its nontrivial root is $\eta_c^{(2)}\simeq 0.87464$, where the other three positivity conditions are not binding. Hence, violation occurs for $\eta>\eta_c^{(2)}$.

 Consequently, the thresholds satisfy
$\eta_c^{(2)}<\eta_y<\eta_c^{(1)}$. For $\eta_c^{(2)}<\eta\leq\eta_y$, both individual bounds are satisfied, while the joint $L=2$ test certifies quantum non-Gaussianity. For example, at $\eta=0.88$, both individual tests pass, but $\lambda_{\min}(R_2)\simeq-0.035885$. The improvement therefore comes from the joint constraint on $\gamma_1$ and $\gamma_2$, demonstrating an advantage over both individual-power tests for this fixed mixer.

\section*{Appendix F: Passive Measurement and Statistical Analysis}

We now describe how to measure the copy-mixing expectations and real moments appearing in the preceding appendices, and how to account for statistical uncertainty when applying the witnesses. Each interferometric shot uses $k$ independently prepared copies of the same $d$-mode state $\rho$ and the single-mode protocol is the case $d=1$. Purity calibration uses separate two-copy shots. Different shots are independent repetitions. All error radii are positive and all failure probabilities lie in $(0,1)$.

\subsection*{F.1 General $k$-copy measurement }
Let $O\in\widetilde G_k$ and choose its unitary spectral decomposition, 
 \begin{equation}
 O=V\operatorname{diag}(e^{i\phi_1},\ldots,e^{i\phi_k})V^\dagger,
 \qquad \phi_j\in\mathbb R.
\end{equation}
In the second-quantization convention,
\begin{equation}
 \Gamma_d(O)=\Gamma_d(V)
 \exp\!\left(i\sum_{r=1}^d\sum_{j=1}^k
 \phi_j\hat n_{r,j}\right)\Gamma_d(V)^\dagger.
\end{equation}
This also covers $\det O=-1$ since an eigenvalue $-1$ corresponds to $\phi_j=\pi$ and contributes photon-number parity.

Apply the passive interferometer $\Gamma_d(V)^\dagger$, which uses the same $V^\dagger$ on the copy ports of each physical mode, and perform photon-number-resolved detection. For the joint outcome $\boldsymbol n=(n_{r,j})_{r,j}$, set
\begin{equation}
 X_O(\boldsymbol n)
 =\exp\!\Big(i\sum_{r=1}^d\sum_{j=1}^k\phi_jn_{r,j}\Big).
\end{equation}
The output probabilities give
\begin{equation}
 \operatorname{Tr}\!\left(\rho^{\otimes k}\Gamma_d(O)\right)
 =\sum_{\boldsymbol n}
 \langle\boldsymbol n|\Gamma_d(V)^\dagger\rho^{\otimes k}
 \Gamma_d(V)|\boldsymbol n\rangle X_O(\boldsymbol n).
\end{equation}
The photon-counting outcome satisfies
\begin{equation}
 \mathbb E[X_O]
 =\operatorname{Tr}\!\left(
 \rho^{\otimes k}\Gamma_d(O)\right).
\end{equation}
Thus, averaging $X_O$ over independent experimental shots gives an unbiased estimator of the copy-mixing expectation. For $d=1$, this expectation reduces to $f_k(\rho,O)$

The real and imaginary parts are estimated by averaging the cosine and sine of the recorded phase. The eigenbasis $V$ determines the optical transformation, while the eigenphases enter only through
classical processing. Phase choices sharing $V$ can therefore be analyzed using the same data; different eigenbases generally require different interferometers.

Because $|X_O|=1$, its real and imaginary parts lie in $[-1,1]$. Hoeffding's inequality and a union bound give
\begin{equation}\label{eq:z-hoeffding}
 \Pr\! \bigg(
\Big|\frac{1}{N_{\mathrm{shot}}}
 \sum_{\ell=1}^{N_{\mathrm{shot}}}X_O(\boldsymbol n_\ell)
 -\operatorname{Tr}\!\left(\rho^{\otimes k}\Gamma_d(O)\right)\Big|
 \geq\epsilon_z\bigg)
 \leq4\exp\!\left(-\frac{N_{\mathrm{shot}}\epsilon_z^2}{4}\right).
\end{equation}
Hence a sufficient integer shot number is
\begin{equation}\label{eq:z-sample-complexity}
 N_{\mathrm{shot}}\geq
 \left\lceil\frac{4}{\epsilon_z^2}
 \log\!\left(\frac4{\delta_z}\right)\right\rceil.
\end{equation}
For $q$ preselected phase choices, replacing $\log(4/\delta_z)$ with $\log(4q/\delta_z)$ gives simultaneous accuracy, even when the estimates use the same shots. These bounds require no photon-number cutoff.

Consider the special case $k=2$ correlated to purity calibration. Let $\mathsf S$ be the SWAP observable exchanging two complete $d$-mode copies, so $\operatorname{Tr}(\rho^{\otimes2}\mathsf S)=p_\rho=\operatorname{Tr}\rho^2$. Each independent measurement produces $Y_t\in\{-1,+1\}$ with $\mathbb E[Y_t]=p_\rho$. It can be realized by a balanced beam splitter between corresponding physical modes of the two copies, followed by parity detection in the relative output ports: $Y_t=(-1)^{\sum_{r=1}^d n_{r,\mathrm{rel},t}}$. Define the raw and clipped estimators by
\begin{equation}\label{eq:purity-estimator-sm}
 \widetilde p=\frac1{M_{\mathrm{shot}}}
 \sum_{t=1}^{M_{\mathrm{shot}}}Y_t,
 \qquad \widehat p=\min\{1,\max\{0,\widetilde p\}\}.
\end{equation}
The raw estimator is unbiased. Clipping need not preserve unbiasedness, but it cannot increase the distance to the true purity $p_\rho\in(0,1]$. Hoeffding's inequality therefore yields
\begin{equation}\label{eq:purity-hoeffding-sm}
 \Pr\bigl(|\widehat p-p_\rho|\geq\epsilon_p\bigr)
 \leq 2\exp\!\left(-\frac{M_{\mathrm{shot}}\epsilon_p^2}{2}\right).
\end{equation}
In particular,
\begin{equation}\label{eq:purity-shot-bound-sm}
 M_{\mathrm{shot}}\geq
 \left\lceil\frac2{\epsilon_p^2}\log\!\left(\frac2{\delta_p}\right)\right\rceil
\end{equation}
gives $|\widehat p-p_\rho|\leq\epsilon_p$ with failure probability at most $\delta_p$. On this same confidence event,
\begin{equation}\label{eq:purity-low-sm}
 p_{\mathrm{low}}=\max\{0,\widehat p-\epsilon_p\}\leq p_\rho.
\end{equation}
For a multimode input this is the global purity of the complete state, rather than a purity of an individual physical mode. 

\subsection*{F.2 Three-copy example and Gaussian-reference test}
Take mode $d=1$ and $k=3$. For the generator $G$ in Eq.~\eqref{eq:G}, use the Fourier matrix
\begin{equation}
 F=\frac1{\sqrt3}
 \begin{pmatrix}
 1&1&1\\
 1&\omega&\omega^2\\
 1&\omega^2&\omega
 \end{pmatrix},
 \qquad \omega=e^{2\pi i/3}.
\end{equation}
With $\phi=\sqrt3\theta$,
\begin{equation}
 \Gamma(e^{\theta G})
 =\Gamma(F)
 \exp\!\left(i\phi
 (\hat n_{3}-\hat n_{2})\right)\Gamma(F)^\dagger.
\end{equation}
An inverse Fourier tritter on the three copies of each physical mode therefore measures $z(\theta)=\operatorname{Tr}\big(\rho^{\otimes3}\Gamma(e^{\theta G})\big)$. Its unbiased estimator is
\begin{equation}
 \widehat z(\theta)
 =\frac1{N_{\mathrm{shot}}}
 \sum_{\ell=1}^{N_{\mathrm{shot}}}
 \exp\!\left(i\phi
 (n_{3,\ell}-n_{2,\ell})\right).
\end{equation}
Only the total count difference between the two relative output ports enters the phase. Since the Fourier eigenbasis is independent of $\theta$, the same count record gives a scan over $\theta$ by classical processing.
The error and shot bounds in Eqs.~\eqref{eq:z-hoeffding} and \eqref{eq:z-sample-complexity} apply directly.

For the single-mode Gaussian-reference test, use the clipped estimator $\widehat p$ from Eq.~\eqref{eq:purity-estimator-sm} and set
\begin{equation}
 \widehat\lambda=\frac{1-\widehat p}{1+\widehat p} .
\end{equation}
No thermal input has to be prepared. The calibration error obeys
\begin{equation}
 |\widehat\lambda-\lambda|
 =\frac{2|\widehat p-p|}{(1+p)(1+\widehat p)}
 \leq2|\widehat p-p|.
\end{equation}
Here $p=p_\rho$ is the true purity and $\lambda=(1-p)/(1+p)$ is the purity-calibrated reference parameter. Because $\widehat p\in[0,1]$, the displayed Lipschitz bound holds on the full sample space. For a desired radius $\epsilon_\lambda$, choose $\epsilon_p=\epsilon_\lambda/2$ in Eq.~\eqref{eq:purity-shot-bound-sm}, equivalently,
\begin{equation}
    M_{\mathrm{shot}}\geq
\lceil8\epsilon_\lambda^{-2}\log(2/\delta_p)\rceil
\end{equation}
suffices. The interval $ I_\lambda=[0,1]\cap  [\widehat\lambda-\epsilon_\lambda,   \widehat\lambda+\epsilon_\lambda]$ then contains $\lambda$ on the purity confidence event. Although every density operator has $p>0$ and hence $\lambda<1$, clipping can give $\widehat\lambda=1$. Define the reference at that endpoint by its one-sided limit,
\begin{equation}
 W_3(1,\theta)=
 \begin{cases}
 0,&\sqrt3\theta\notin2\pi\mathbb Z,\\
 1,&\sqrt3\theta\in2\pi\mathbb Z.
 \end{cases}
\end{equation}
This is a limiting reference value, where $\lambda=1$ does not define a normalizable thermal state. For nontrivial angles,
\begin{equation}
 \frac{\partial W_3}{\partial\lambda}(\lambda,\theta)
 =-\frac{2(1-\lambda^2)(1-\cos(\sqrt3\theta))}
 {(1-2\lambda\cos(\sqrt3\theta)+\lambda^2)^2}.
\end{equation}
At trivial angles the extended reference is identically one and its derivative is zero. To propagate calibration uncertainty \cite{taylor1997error}, set
\begin{equation}
 L_\lambda(\theta)=\sup_{\mu\in I_\lambda}
 \left|\frac{\partial W_3}{\partial\lambda}(\mu,\theta)\right|.
\end{equation}
This is finite for each fixed angle, and the mean-value bound gives $|W_3(\widehat\lambda,\theta)-W_3(\lambda,\theta)|\leq L_\lambda(\theta)\epsilon_\lambda$ on the purity event. Consequently,
\begin{equation}\label{eq:finite-sample-witness}
 |\widehat z(\theta)-W_3(\widehat\lambda,\theta)|
 >\epsilon_z+L_\lambda(\theta)\epsilon_\lambda
\end{equation}
certifies non-Gaussianity of a single-mode input with confidence at least $1-\delta_z-\delta_p$. 

The union bound does not require the two estimators to be independent, and the phase and purity measurements consume $3N_{\mathrm{shot}}+2M_{\mathrm{shot}}$ copies of $\rho$.

\subsection*{F.3 Quantum non-Gaussianity tests}
The same interferometric data and purity calibration also give the real-moment tests for the Gaussian convex hull. For fixed $k\geq2$, $O\in\widetilde G_k$, and a preselected level $L\geq1$, the moments $\gamma_s=\operatorname{Re}\operatorname{Tr} [\rho^{\otimes k}\Gamma_d(O)^s]$ have unbiased estimators
\begin{equation}\label{eq:moment-estimator-sm}
 \widehat\gamma_s
 =\frac1{N_{\mathrm{shot}}}
 \sum_{\ell=1}^{N_{\mathrm{shot}}}
 \cos\!\left(s\sum_{r=1}^d\sum_{j=1}^k
 \phi_jn_{r,j,\ell}\right),
 \quad s=1,\ldots,L,\qquad
 \gamma_0=\widehat\gamma_0=1.
\end{equation}
Changing $s$ requires only classical processing of the same count record.

For the three-copy tritter, $\widehat\gamma_s=\operatorname{Re}\widehat z(s\theta)$. With $C_L=(\gamma_{|j-i|})_{i,j=0}^L$ and $J_L=\boldsymbol1\boldsymbol1^T$, Gaussian mixtures must satisfy $C_L-p^{k-1/2}J_L\succeq0$.

Each cosine lies in $[-1,1]$. Hoeffding's inequality and a union bound therefore give
\begin{equation}\label{eq:moment-hoeffding-sm}
 \Pr\!\left(\max_{1\leq s\leq L}
 |\widehat\gamma_s-\gamma_s|\geq\epsilon_\gamma\right)
 \leq2L\exp\!\left(
 -\frac{N_{\mathrm{shot}}\epsilon_\gamma^2}{2}\right).
\end{equation}
A sufficient shot count for simultaneous accuracy is
\begin{equation}\label{eq:moment-sample-complexity-sm}
 N_{\mathrm{shot}}\geq
 \left\lceil\frac{2}{\epsilon_\gamma^2}
 \log\!\left(\frac{2L}{\delta_\gamma}\right)\right\rceil.
\end{equation}
Correlations between moment estimates from the same shots do not affect this bound. If settings or hierarchy levels are selected from the data, the confidence event must cover that selection, for example through a preselected finite family and its union bound.

Use the global purity lower bound $p_{\mathrm{low}}$ from Eq.~\eqref{eq:purity-low-sm} and form $\widehat C_L=(\widehat\gamma_{|j-i|})_{i,j=0}^L$. On the simultaneous moment event, $\|\widehat C_L-C_L\|\leq L\epsilon_\gamma$ because its diagonal is exact. Every Gaussian mixture satisfies $C_L-p_{\mathrm{low}}^{\,k-1/2}J_L\succeq0$ on the purity event. Hence the sufficient condition
\begin{equation}
 \lambda_{\min}\!\left(
 \widehat C_L-p_{\mathrm{low}}^{\,k-1/2}J_L\right)
 <-L\epsilon_\gamma
\end{equation}
certifies quantum non-Gaussianity with confidence at least $1-\delta_\gamma-\delta_p$.
The purity lower bound accounts conservatively for calibration error, while the eigenvalue margin accounts for moment errors. The phase and SWAP measurements consume $kN_{\mathrm{shot}}+2M_{\mathrm{shot}}$ full copies of $\rho$, including $3N_{\mathrm{shot}}+2M_{\mathrm{shot}}$ for the tritter.

These guarantees assume calibrated passive transformations and photon-number-resolved detection. Known absolute measurement-bias bounds can be added to the corresponding statistical radii before
applying the certificates. Identical input loss may instead be absorbed into the common state $\rho$, provided the purity calibration measures that same state.

\bibliographystyle{apsrev4-2}
\bibliography{1bibNONGaussian}
\end{document}